\documentclass[a4paper,twocolumn,11pt,accepted=2025-06-25]{quantumarticle}
\pdfoutput=1
\usepackage[colorlinks=true, citecolor=blue, urlcolor=blue,linkcolor=red]{hyperref}
\usepackage[utf8]{inputenc}
\usepackage[english]{babel}
\usepackage[T1]{fontenc}
\usepackage{xkeyval, etoolbox, geometry, fancyhdr,ltxgrid,ltxcmds,type1ec}
\usepackage{tikz}
\usepackage{lipsum}
\usepackage{amsmath,bm,subfigure,physics,mathtools,amsfonts,amsthm,comment,amssymb}
\usepackage[normalem]{ulem}
\usepackage[thinc]{esdiff}
\usepackage{xcolor}
\usepackage{dsfont}
\usepackage{float}
\usepackage{tabularx}
\usepackage{array}
\usepackage[export]{adjustbox}
\usepackage{mathtools}
\usepackage{IEEEtrantools}
\usepackage[numbers,sort&compress]{natbib}

\newtheorem{theorem}{Theorem}

\newtheorem{corollary}[theorem]{Corollary}

\newcommand{\ph}[1]{{\color{magenta}#1}}
\newcommand{\ga}[1]{{\color{blue}#1}}

\begin{document}
\title{Indefinite Time Directed Quantum Metrology} 

\author{Gaurang Agrawal}
\affiliation{Indian Institute of Science Education and Research, Homi Bhabha Rd, Pashan, Pune 411 008, India}
\affiliation{Harish-Chandra Research Institute, A CI of Homi Bhabha National Institute,  Chhatnag Road, Jhunsi, Allahabad - 211019, India}
\orcid{0009-0006-2072-4542}
\author{Pritam Halder}
\affiliation{Harish-Chandra Research Institute, A CI of Homi Bhabha National Institute,  Chhatnag Road, Jhunsi, Allahabad - 211019, India}
\orcid{0000-0002-2409-3630}
\author{Aditi Sen(De)}
\affiliation{Harish-Chandra Research Institute, A CI of Homi Bhabha National Institute,  Chhatnag Road, Jhunsi, Allahabad - 211019, India}
\orcid{0000-0003-1693-0440}

\maketitle

\begin{abstract}
   We explore the performance of the metrology scheme by employing a quantum time flip during encoding, a specific case of processes with indefinite time direction, which we refer to as indefinite time directed metrology (ITDM). 
   In the case of single parameter estimation of a unitary, we demonstrate that our protocol can achieve Heisenberg scaling $(1/N)$ with product probe states, surpassing the standard quantum limit \((1/\sqrt{N})\),  where $N$ is the number of particles in the probe. We establish this by computing the quantum Fisher information (QFI) which is a lower bound on the root mean square error occurred during parameter estimation.
   Although we analytically prove the optimality of the symmetric product probe state in ITDM, entangled probe states produce a higher QFI than optimal product probes without enhancing scaling, highlighting the non-essentiality of entanglement. For phase estimation, we propose a single-qubit measurement on the control qubit that accomplishes near-optimal Fisher information and eventually reaches Heisenberg scaling.  Our findings reveal the best orientation of product probe states in every pertinent situation, emphasizing its independence from the parameter to be estimated \textcolor{black}{in the limiting case.} 
  Furthermore, we illustrate the benefits of ITDM in noisy metrology, outperforming existing techniques in some situations.

\end{abstract}

\section{Introduction}

The goal of quantum metrology \cite{Giovannetti2006Jan,paris2009quantum,Giovannetti2011Apr,Toth2014Oct,Degen2017Jul, Pezze2018Sep,Pirandola2018Dec} 
is to make extremely accurate estimates of a system's properties which is crucial in diverse scenarios such as magnetometry \cite{Jones2009Apr,Wasilewski2010Mar} 
 quantum imaging \cite{Perez-Delgado2012Sep, Genovese2016Jun, Moreau2019Jun}, atomic clock \cite{Appel2009Jul,Katori2011Apr,Ludlow2015Jun,Nichol2022Sep}, biological sensing
\cite{Taylor2016Feb, Mauranyapin2017Aug}, dark matter search \cite{Malnou2019May,Bass2024May}, gravitational wave detection \cite{Caves1981Apr,Schnabel2010Nov, BibEntry2011Dec}. 
In particular, the metrological procedure involves measuring a quantum state that encodes the parameter to be estimated, and its accuracy is typically quantified by the root mean square error \cite{Giovannetti2011Apr}.   Under plausible assumptions, the quantity is constrained below by Fisher information (FI), also known as the Cramer-Rao bound \cite{Cramr1946}, while optimizing overall measurement procedures provides quantum Fisher information (QFI), which specifies the ultimate precision limit \cite{Braunstein1994May}. 
%The metrological protocol involves measuring a quantum state that encodes the parameter of interest, to obtain information about its value. An unbiased estimator processes measurement data to estimate the parameter with a precision bounded by the Cramér-Rao inequality \cite{Braunstein1994May}, which sets a lower limit on the variance of the estimate by the inverse of the classical Fisher information (FI). Optimizing the FI over all measurement strategies yields the quantum Fisher information (QFI), which defines the ultimate precision limit and saturates the Cramér-Rao bound. 
It was demonstrated that, while uncorrelated probes yield QFI that scales only linearly with the number of subsystems, achieving the standard quantum limit (SQL), nonclassical resources such as entanglement or spin-squeezing can improve sensitivity by exceeding linear growth to quadratic or beyond, referred to as the Heisenberg limit (HL) or super-Heisenberg precision \cite{Wineland1992Dec, Leibfried2004Jun, Pezze2009Mar, Giovannetti2011Apr, Toth2014Oct, Augusiak2016Jul}. 
%Uncorrelated probe states yield QFI that scales at most linearly with the number of subsystems, achieving the standard quantum limit (SQL). However, leveraging nonclassical resources such as entanglement or spin-squeezing \cite{Wineland1992Dec, Leibfried2004Jun, Pezze2009Mar, Giovannetti2011Apr, Toth2014Oct, Augusiak2016Jul} can enhance sensitivity, enabling QFI scaling to surpass linear growth and reach the Heisenberg limit (HL) or even super-Heisenberg precision. 

In conventional metrological protocols, the encoding operation is performed within a  framework that adheres to a well-defined causal order and a fixed temporal direction.  However, recent studies \cite{zhao2020quantum, liu2023optimal, chapeau2021noisy, Frey2019Apr, Ban2023Apr, mothe2024reassessing, an2024noisy, Yin2023Aug, Liu2024Dec} have explored metrological enhancements of  HL and beyond with the aid of indefinite causal order (ICO) \cite{Oreshkov2012Oct, Araujo2015Oct, Procopio2015Aug, Rubino2017Mar, Goswami2018Aug, Stromberg2023Aug, Rozema2024Aug} in which 
the quantum switch \cite{chiribella2013quantum} serves as an example of a quantum-controlled superposition of several causally ordered transformations \cite{Wechs2021Aug}.
%quantum-controlled superposition of multiple causally ordered transformations \cite{Wechs2021Aug}, exemplified by the quantum switch \cite{chiribella2013quantum}.
Importantly, 
%Experimental implementations of such protocols have also been reported 
such protocols have already been implemented in laboratories \cite{Yin2023Aug,an2024noisy}. Further, beyond enhancing sensitivity, ICO has exhibited benefits across various domains including state and channel discrimination \cite{Chiribella2012Oct,Bavaresco2021Nov,Bavaresco2022Apr,Kechrimparis2024Dec}, communication capacity and complexity \cite{ebler2018enhanced,Wei2019Mar,Salek2018Sep,wu2024general,Kristjansson2020Jul,Wilson2020Mar,Loizeau2020Jan,Guo2020Jan,Goswami2020Aug,Chiribella2021Nov,Sazim2021Jun}, along with thermodynamical tasks including charging, refrigeration, work extraction and thermal state activation  \cite{Felce2020Aug,Guha2020Sep,Simonov2022Aug,Simonov2022Mar,Francica2022,Nie2022Sep,Zhu2023Dec}. 

%Notably, indefinite causal ordered processes represent a particular kind of higher-order quantum transformation that inputs or outputs a quantum channel or other quantum operations \cite{Bisio2019}. Lately, 
Similar to the concept of ICO, 
a new broader class of higher-order operations \cite{Bisio2019} has recently been discovered, termed quantum operations with indefinite time direction or indefinite input-output direction \cite{chiribella2022quantum}. These operations lack a fixed temporal direction and involve a coherent superposition of forward and backward time directions. In the backward time direction, the roles of input and output port are interchanged, defining  quantum processes which receive the input in the future and output in the past timeline. Indefinite time direction can be realized in bidirectional quantum devices \cite{chiribella2022quantum, Liu2023Apr} which include crystals that rotate the polarization of single photons such as half-wave plates, and quarter-wave plates.  Analogous to the quantum switch in the context of ICO, the quantum ``time flip'' (TF) \cite{chiribella2022quantum, Liu2023Apr}, a paradigmatic example of indefinite time directed processes, has shown promising advantages in certain operator discrimination tasks \cite{chiribella2022quantum}, communication capacity \cite{Liu2023Apr}, emergence of dynamical memory effect in memoryless phase-covariant processes \cite{karpat2024memory}, and also answering foundational queries such as the thermodynamic “arrow of time” and work-extraction \cite{Rubino2021Nov, Rubino2022Mar}. 
%\ph{thermodynamic arrow of time }
Furthermore, its feasibility has been established through experimental implementations \cite{stromberg2024experimental, guo2024experimental}. %Additionally, foundational queries such as the thermodynamic “arrow of time” and thermodynamic work in the quantum superpositions of forward and backward time evolution  have been investigated. 
Despite these advancements, the potential applications of the quantum TF in quantum technologies remain largely unexplored (although see a recent work which uses   quantum time flip \cite{Xia2024Jul} in continuous variable systems for estimating optical phase and light-beam rotation  with  photon number and orbital angular momentum  being resources respectively). 

%with photon number and orbital angular momentum 
 
%Recently, Xia et al. \cite{Xia2024Jul} have experimentally demonstrated Heisenberg-limited precision in optical phase and light-beam rotation estimation using quantum time-flip with photon number and orbital angular momentum as respective resources. However, to date, no studies have explored 

In this paper, we bridge the gap by introducing a novel framework for quantum metrology 
%the usage of TF in discrete-variable regimes, where the fundamental resource is the number of subsystems in the probe system for enhancing precision.
%In this work, we 
%introducing a novel framework 
that consists of encoding operations with indefinite input-output direction. By exploiting the properties of quantum TF, we devise a protocol that enables Heisenberg-limited precision in the estimation of an arbitrary single parameter of unitary operator, which we refer to as indefinite time directed metrology (ITDM). We prove that 
initializing the probe as pure product states in which all subsystems align in a specific direction, known as a symmetric product probe state, yields the maximum QFI  using a time-flipped encoding approach.
%initializing the probe as pure product states where all the subsystems align in the same direction, called a symmetric product probe states attains the maximum value of QFI with time-flipped encoding procedure. 
%Specifically, without preparing entangled probe states, the optimal symmetric product probe states can attain Heisenberg-limited precision.
%Consequently, we establish the fact that 
Further, we demonstrate that entanglement can yield a higher QFI  than the product probes, but it cannot offer a scaling benefit beyond the HL in this encoding procedure.
%entanglement is not a prerequisite for attaining HL precision and cannot provide a scaling advantage beyond the HL  when the temporal ordering of the encoding process is rendered indefinite through quantum control. 
%
We also extract requirements for encoded state, in which ITDM cannot outperform SQL -- we call the no-advantage condition (NAC).
%We also derive the condition on encoded state which we call the no-advantage condition in which ITDM cannot surpass the SQL. 

%and quantitatively assess the advantage conferred by time-flipping in our study.

We  demonstrate the benefit of TF encoding for estimating phase-shift and axis  of a unitary. Specifically, in the case of phase estimation, we devise a measurement strategy that operates exclusively on the control qubit yielding an FI that exhibits Heisenberg scaling. Interestingly, we exhibit that the parameter dependency of the optimal input probe state and FI can be eliminated when the parameter to be evaluated is small. Furthermore, 
%in the case of a single qubit noisy ITDM (NITDM) where 
when noise influences the encoding strategy and input probes, our protocol showcases the advantage of the time-flipped strategy over the existing quantum ``switched'' strategy \cite{chapeau2021noisy} and regular strategy within specific parameter values and noise regimes.  We also briefly analyze the effect of noise in multi qubit noisy ITDM and present a possible way to get a higher FI than the standard method, emphasizing the need for further investigation in this area.

%Additionally, we provide a brief discussion on the impact of noise in multi qubit NITDM and explore a potential strategy for achieving a higher FI compared to the regular approach, highlighting the need for further investigation in this direction.

The  paper is organized as follows. In Sec.~\ref{sec:ITDM_framework}, we introduce the strategy based on time-flip encoding and derive the Heisenberg scaling for the general case. We present the optimality of symmetric product probe states among all product states and NAC condition with a measure of advantage in Sec.~\ref{sec:product_probe_itdm}. In Sec.~\ref{sec:phase_estimation_itdm}, we provide the applications of our procedure to phase estimation. %Sec.~\ref{sec:control_qubit_measure} contains the analysis of our  with fixed projective measurements in control qubit. 
We evaluate the use of entangled probes in Sec.~\ref{sec:entanglement} and demonstrate that entangled states cannot provide a higher QFI scaling than the product probes in this set-up. In Sec.~\ref{sec:noisy_adv}, we investigate the advantages of our scheme under noise and compare them with existing protocols. Finally, the results are summarized in Sec.~\ref{sec:conclusion}.

\section{Framework for Indefinite Time Directed Metrology}
\label{sec:ITDM_framework}

Quantum time flip is defined on the set of bidirectional quantum processes which can operate either in forward or backward time directions. A quantum channel \(\mathcal{C}\) is referred to as \textit{bidirectional} if \(\Omega(\mathcal{C})\) is a completely positive and trace-preserving (CPTP) map where \(\Omega\) denotes the \textit{input-output inversion map}. As demanded in Ref.~\cite{chiribella2022quantum}, this map $\Omega$ should have certain properties including ``order reversing'', ``identity preserving'', $\Omega(\mathcal{A}) \neq \Omega(\mathcal{B})$ when channels $\mathcal A, \mathcal B$  are unequal, and $\Omega\big(q \mathcal{A} + (1-q)\mathcal{B}\big) = q\Omega(\mathcal{A}) + (1-q)\Omega(\mathcal{B})$ with $\  1\geq  q \geq 0$. Moreover, the Kraus operators $\{C_i\}$ of $\mathcal{C}(*) = \sum_i C_i (*) C_i^\dagger$ satisfy bistochasticity conditions, i.e., $\sum_i C^\dagger_iC_i=\sum_iC_iC_i^\dagger=I$, with $I$ being the identity matrix on the Hilbert space of the system. The only two possibilities of $\Omega$ are $\Omega(\mathcal C)(*)=\mathcal C^\dagger(*)=\sum_iC_i^\dagger(*)C_i$ or $\Omega(\mathcal C)(*)=\mathcal C^T(*)=\sum_iC_i^T(*)C_i^{T^\dagger}$, up to unitary equivalence. $\mathcal C$ as the input of the quantum time flip supermap generates the  time flipped channel, denoted as $T_{\mathcal{C}}$. The operation of  $T_{\mathcal{C}}$ is described by  \(T_{\mathcal{C}}(\rho_c \otimes \rho) = \sum_{i} T_{i}^{\mathcal C}(\rho_c \otimes \rho) T_{i}^{\mathcal C\dagger}\) where
$T_{i}^{\mathcal{C}} = \ket{0}\bra{0} \otimes C_i + \ket{1}\bra{1} \otimes \Omega(C_i)$ are the Kraus operators and the first register functions as the control for the time direction. We perform encoding in parameter estimation protocol with the aid of indefinite time ordering, defined as ITDM.  

\begin{figure}[h]
    \centering
    \includegraphics[width=0.55\textwidth]{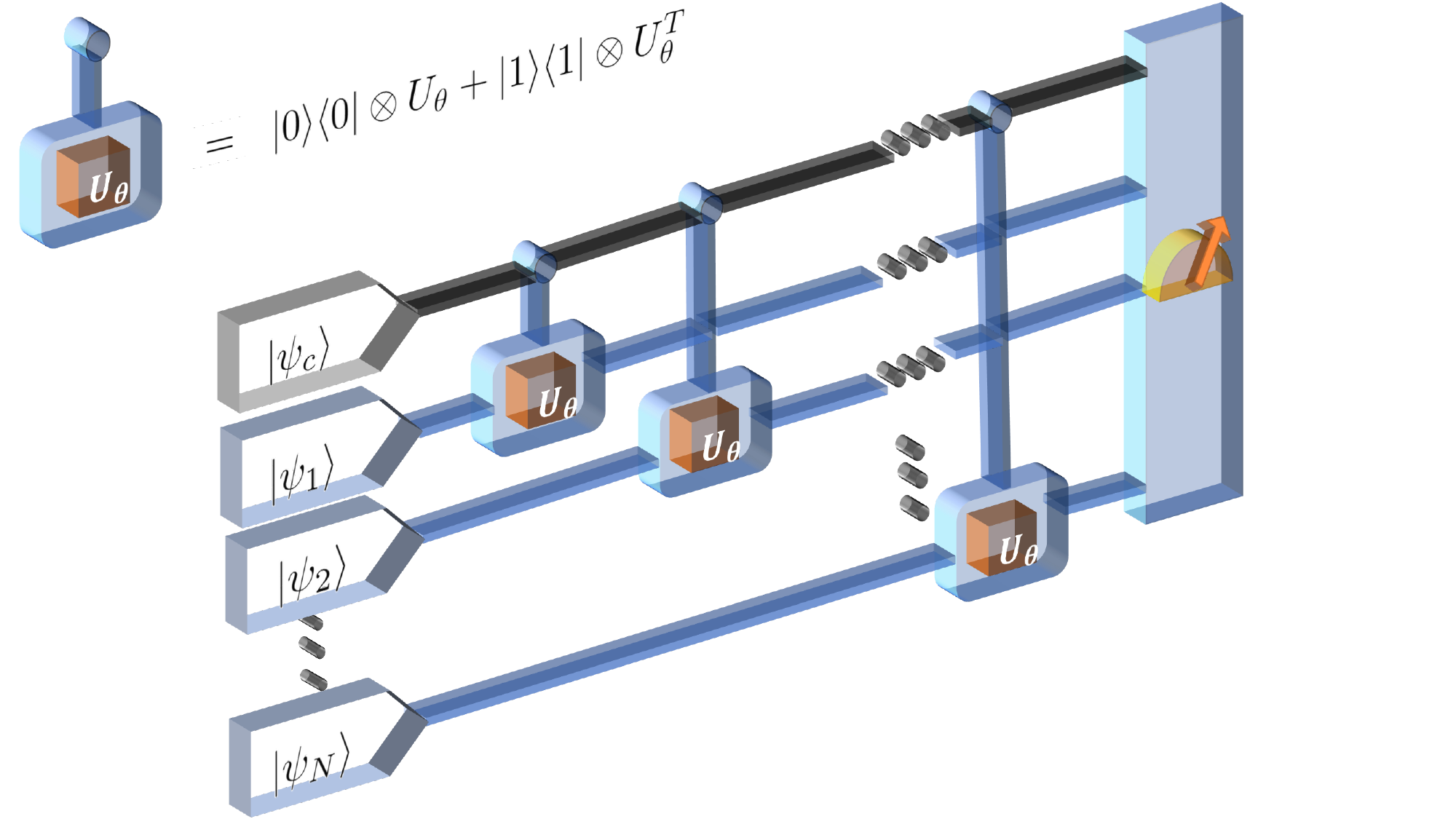}
    \caption{\textbf{ITDM scheme.} Taking $N$-qubit pure product probe states, $\{\ket{\psi_i}\}^N_{i=1}$ along with a control qubit, $\ket{\psi_c}$ we encode the parameter $\theta$ of $U_\theta$ on the probe state using time flipped gate, $T_{U_\theta}=|0\rangle\langle0| \otimes U_\theta + |1\rangle\langle1| \otimes U_\theta^{T}$. Therefore, the whole encoding procedure is performed in a single time step with $N$ queries of $T_{U_\theta}$, given by $\mathbb{T}_{U_\theta}=\prod_{i=1}^N (T^{i}_{U_\theta}\otimes\mathbb{I}_{N-1}^{\neq i})= \ketbra{0}{0}\otimes U_\theta^{\otimes N}+\ketbra{1}{1}\otimes {U_\theta^T}^{\otimes N}$ with $\mathbb{I}_{N-1}^{\neq i}=\bigotimes_{\substack{j=1\\j\neq i}}^N{I}_j$. This encoded state is then measured either locally on the control qubit or globally on the basis of symmetric logarithmic derivative  operator to obtain the information about $\theta$.}
    \label{fig:strat}
\end{figure}

\subsection{Time-flipped encoding strategy}

In quantum estimation theory, given a unitary gate $U$  which is a function of $\theta$, i.e., $U\equiv U_\theta$, our task is to estimate $\theta$.  Note that the parameter $\theta$ can be directly related to a physical quantity of interest, such as the strength of a magnetic field,  phase shift acquired by a photon's polarization state in a photonic setup.  Instead of conventional unitary encoding schemes, we employ TF-encoding strategy comprising of $N$ number of bipartite  gates, $T_{U_\theta}=|0\rangle\langle0| \otimes U_\theta + |1\rangle\langle1| \otimes U_\theta^{T}$, acquired by applying the TF-supermap on  $U_\theta$, on $N$-qubit input probe, $\ket{\psi^N}$, assisted by a control qubit, $\ket{\psi_c}= \sqrt{p_c} \ket{0} + e^{i \theta_c} \sqrt{1 - p_c} \ket{1}$  with $0\leq p_c\leq1$ and $0\leq\theta_c\leq 2\pi$ (see Fig.~\ref{fig:strat}).
% Let us prepare the input probe into a pure product state
% \begin{eqnarray}
%     \ket{\Psi}=\bigotimes_{i=1}^{N} \ket{\psi_i}.
%     \label{eq:i_prod}
% \end{eqnarray}  
% For a given unitary $U_\theta$, the ``time-flipped (TF)''  strategy generates the gate operation $T_{U_\theta}^{ci}(\ket{\psi_c}\otimes\ket{\psi_i})=\left( |0\rangle\langle0| \otimes U_\theta + |1\rangle\langle1| \otimes U_\theta^{T}\right)_{ci} (\ket{\psi_c}\otimes\ket{\psi_i})$ where $\ket{\psi_c}= \sqrt{p_c} \ket{0} + e^{i \theta_c} \sqrt{1 - p_c} \ket{1}$ denotes the quantum state of the control qubit  with $0\leq p_c\leq1$ and $0\leq\theta_c\leq 2\pi$. According to Fig.~\ref{fig:strat}, we encode $\theta$ by applying multiple $T_{U_\theta}^{ci}$ gates.
Mathematically, the state after encoding operation  can be written as $\ket{\Phi} = \mathbb{T}_{U_\theta}\ket \Psi = \mathbb{T}_{U_\theta} \ket{\psi_c} \otimes\ket{\psi^N} $, where
\begin{IEEEeqnarray}{rCl}
   \nonumber \mathbb{T}_{U_\theta}&&=\prod_{i=1}^N (T^{i}_{U_\theta}\otimes\mathbb{I}_{N-1}^{\neq i}),\\
   &&= \prod_{i = 1}^N \bigg(|0\rangle\langle0| \otimes I^{\otimes i-1} \otimes U_\theta  \otimes I^{\otimes N-i} \nonumber \\
   &&\nonumber \qquad \quad + |1\rangle\langle1| \otimes I^{\otimes i-1} \otimes U_\theta^T  \otimes I^{\otimes N-i} \bigg),\\ &&=\ketbra{0}{0}\otimes U_\theta^{\otimes N}+\ketbra{1}{1}\otimes {U_\theta^T}^{\otimes N}, 
\label{eq:time_flip_gate}
\end{IEEEeqnarray}
with $\mathbb{I}_{N-1}^{\neq i}=\bigotimes_{\substack{j=1\\j\neq i}}^N{I}_j$. Specifically, when $\ket{\psi^N}$ is fully separable, i.e., $\ket{\psi^N} = \bigotimes_{i=1}^N\ket{\psi_i}$, the encoded state  can be written as
\begin{equation}
\ket{\Phi}
= \sqrt{p_c} \ket{0} \bigotimes_{i=1}^{N} \ket{\chi_i} + e^{i \theta_c} \sqrt{1 - p_c} \ket{1} \bigotimes_{i=1}^{N} \ket{\omega_i},
\label{eq:PureStateOutput}
\end{equation}
where $U_{\theta} \ket{\psi_i} = \ket{\chi_i},$ $\ U_{\theta}^{T} \ket{\psi_i} = \ket{\omega_i}$. Subsequently, we analyze the performance of the estimation protocol with product probe states and compare its performance with the entangled ones.

\subsection{Quantum Fisher information in ITDM}
\label{sec:product_probe_itdm}

 The goal here is to estimate the parameter $\theta$ involved in $U_\theta$ with a high precision which can be fulfilled by performing measurements on the encoded state. Importantly, for any kind of measurements, the precision in estimating $\theta$ with an unbiased estimator, quantified by the variance of $\theta$, is bounded below by the Cramer-Rao bound via  Fisher information, $\mathcal F_\theta$ as $(\Delta\theta)^2\geq (\lambda \mathcal F_\theta)^{-1}$, where $\lambda$ denotes the number of measurements available \cite{Braunstein1994May, paris2009quantum}. In particular, when the decoding measurement is performed in the basis of symmetric logarithmic derivative (SLD), Cramer-Rao bound is saturated and FI is termed as quantum Fisher information. The QFI for a pure state $\ket{\Xi}$ can be written as 
\begin{equation}
\mathcal{Q}_\theta= 4 \left( \langle \dot{\Xi} | \dot{\Xi} \rangle - \left|\langle{\dot{\Xi}}|\Xi \rangle \right|^2 \right),
\label{eqn: QuantumFisherInformationPureState1}
\end{equation}
where $\dot{\ket{\Xi}}=\diff{\ket{\Xi}}{\theta}$. For $N$ number of uncorrelated particles, $(\Delta \theta)^2$ is known to scale as $1/N$, referred to as the standard quantum limit (SQL) while with the help of quantum entanglement, it is possible to diminish the error as $(\Delta \theta)^2 \sim 1/{N^2}$ \cite{paris2009quantum,Toth2014Oct,Degen2017Jul,Giovannetti2011Apr}, known as the Heisenberg limit (HL). It will be interesting to find that when encoding is performed with TF operator, how quantum Fisher information scales with $N$. In our case of product probes in Eq.~\eqref{eq:PureStateOutput}, QFI takes the form as (see Appendix.~\ref{app:QFI_calc})
\begin{eqnarray}
    \nonumber\mathcal{Q} &=& 4 \bigg[ p_c \left( \sum_{k=1}^N |\dot{\chi}_k| + \sum_{\substack{j,k=1,\\j \neq k}}^N \alpha_k \alpha_j^{*} \right) \\ \nonumber &&+ (1 - p_c) \left( \sum_{k=1}^N |\dot{\omega_k}| + \sum_{\substack{j,k=1,\\j \neq k}}^N \beta_k \beta_j^{*} \right) \nonumber\\
    &&- \bigg| p_c \sum_{k=1}^N \alpha_k + (1 - p_c) \sum_{k=1}^N \beta_k \bigg|^2 \bigg], \label{eq:symmproveeqnQ0}
\end{eqnarray}
where
\begin{align}
 & \quad \langle \dot{\chi}_k | \chi_k \rangle = \alpha_{k}, \quad \langle \dot{\chi}_k | \dot{\chi}_k \rangle = | \dot{\chi}_k |, \notag\\
& \quad \langle \dot{\omega}_k | \omega_k \rangle = \beta_{k}, \quad \langle \dot{\omega_k} | \dot{\omega_k} \rangle = |\dot{\omega_k} |. \label{eqn:Singleparamsld_labels}
\end{align}
To analyze the scaling of QFI, let us concentrate on the simplest scenario where all the qubits in the probe state are equal, i.e., $\ket{\psi_i} = \ket{\psi}$ $\forall i$, leading to a $N$-qubit symmetric probe state while for the asymmetric one, $\ket{\psi_i} \neq \ket{\psi} ~ \forall i$ . Therefore, in the symmetric scenario, the QFI reduces to
\begin{IEEEeqnarray}{rCl}
&\mathcal{Q}^S& \nonumber \\&=& 4 N^2 \Big[ p_c |\alpha|^2 + (1 - p_c) |\beta|^2 - |p_c \alpha + (1 - p_c) \beta|^2 \Big] \notag \nonumber \\
&&+ 4 N \Big[ p_c \Big(| \dot{\chi} | - |\alpha|^2 \Big) + (1 - p_c) \Big( | \dot{\omega} | - |\beta|^2 \Big) \Big], \nonumber\\ 
&=& 4 N^2 p_c(1-p_c)|\alpha-\beta|^2\nonumber\\
&&+ 4 N \Big[ p_c \Big(| \dot{\chi} | - |\alpha|^2 \Big) + (1 - p_c) \Big( | \dot{\omega} | - |\beta|^2 \Big) \Big],
\label{eq:QFI_singleparam_Unitary_SLD0}
\end{IEEEeqnarray}
with $\alpha_k = \alpha, \ \beta_k = \beta, \ \ |\dot\chi_k| = |\dot\chi|,$ and $|\dot\omega_k| = |\dot\omega|$ $\forall k$. From the above expression, it is clear that $\mathcal{Q}^S \sim N^2$ as $N\to\infty$, which implies that by using symmetric product state $\ket{\Psi^S} = \ket{\psi_c}\otimes\ket{\psi}^{\otimes N}$ and encoding $\theta$ by the time-flipped operation which comprises multiple two-qubit controlled gates (see Fig.~\ref{fig:strat}), the precision limit of $\theta$ can go beyond SQL and can achieve HL. 

Let us now inquire whether we can beat SQL for all kinds of  parameterized unitaries. We find that this is not the case, specifically, ITDM can not surpass the standard quantum limit when $\alpha=\beta$. This condition, referred to as the ``no advantage condition'' (NAC), simplifies to $(\partial_\theta U_\theta)^\dagger U_\theta = (\partial_\theta U_\theta^T)^\dagger U_\theta^T$. If $U_\theta$ satisfies NAC, it implies that there do not exist any probe and control state for which  one can go beyond the SQL. Notably, the NAC is satisfied when $U^T=\pm U$, although, additional solutions for $U$ may also exist. 
 
Motivated by no-advantage condition, we propose a measure of advantage ($M_A$) in our protocol as the value of coefficient of $N^2$ in Eq.~\eqref{eq:QFI_singleparam_Unitary_SLD0}. It turns out to be $M_A=4p_c(1-p_c)\max_{\ket{\psi}}\big|\big\langle(\partial_\theta U_\theta)^\dagger U_\theta - (\partial_\theta U_\theta^T)^\dagger U_\theta^T\big\rangle\big|^2 = 4p_c(1-p_c)|\lambda_{\max}|^2$ where $\lambda_{\max}$ is the eigenvalue of $(\partial_\theta U_\theta)^\dagger U_\theta - (\partial_\theta U_\theta^T)^\dagger U_\theta^T$ with the biggest modulus. There is no advantage when either $p_c = 0 \text{ and }1,$ i.e., when the time flip does not work or when NAC is satisfied where $|\lambda_{\max}|^2 = 0$ thereby, resulting in a vanishing measure. On the other hand, when the measure is non-zero, we can, in principle, set our probe to the eigenvector corresponding to $\lambda_{\max}$ and gain QFI with an $N^24p_c(1-p_c)|\lambda_{\max}|^2$.

It is now natural to ask -- \textit{Is the symmetric product probe state optimal for ITDM}? In other words, if we can exhibit that one cannot achieve higher scaling with asymmetric product probes than the one obtained with symmetric states, the query above is resolved. Indeed, we answer this question positively in the following Theorem.

\begin{theorem}
    Among all pure product states, given an asymmetric state, there exists a symmetric state which provides a higher value of QFI in ITDM.  
    \label{th1}
\end{theorem}
\begin{proof}
    Given an asymmetric product probe state $\ket{\Psi}=\ket{\psi_c}\bigotimes_{i=1}^N\ket{\psi_i}$ with QFI $\mathcal Q^A$, we denote the set of $N$ symmetric product states as $\mathcal{S}_{\Psi}:=\{\ket{\Psi^S_l}\}_{l=1}^N$. Here, the state $\ket{\Psi^S_l}=\ket{\psi_c}\ket{\psi_l}^{\otimes N}$ in this set have QFI $\mathcal{Q}^S_l$. 
    % The QFI for the, asymmetric probe state $\ket{\Psi}$ is labeled as $\mathcal{Q}^A$.
    % Given an arbitrary  product state $\ket{\Psi}=\ket{\psi_c}\bigotimes_{i=1}^N\ket{\psi_i}$, we denote the set of $N$ symmetric product states as $\mathcal{S}_{\Psi}:=\{\ket{\Psi^S_l}\}_{l=1}^N$. Here, the state $\ket{\Psi^S_l}=\ket{\psi_c}\ket{\psi_l}^{\otimes N}$ in this set have QFI $\mathcal{Q}^S_l$. We denote the QFI for the, asymmetric probe state $\ket{\Psi}$ as $\mathcal{Q}^A$.
    Without loss of generality, we assume that the optimal probe state among all symmetric product states in $\mathcal{S}_\Psi$ is $\ket{\Psi^S_1}$. Therefore, mathematically, we can write $\mathcal{Q}^S_1\geq\mathcal{Q}^S_l \quad \forall l\in [N],$ implying $ \mathcal{Q}^S_1-\mathcal{Q}_{avg}\geq 0,$
    where $\mathcal{Q}_{avg}=\frac{1}{N}\sum_{l=1}^N\mathcal{Q}^S_l$. Let us now consider the quantity, $\mathcal{Q}_{avg}-\mathcal{Q}^A$. On the other hand, it can be shown that (see Appendix.~\ref{ap:th1}) 
    \begin{align}
        \nonumber &\mathcal{Q}_{avg}-\mathcal{Q}^A\\ \nonumber
        &=4 p_c (1 - p_c) \left(N\sum_{k =1}^N \left| \alpha_k -\beta_k \right|^2 -  \left| \sum_{k =1}^N \alpha_k -\beta_k\right|^2\right), \\&\geq 0.
        \label{eq:avg_anti}
    \end{align}
    Therefore, we get $\mathcal{Q}_1^S\geq\mathcal{Q}^A$. Hence, the theorem is proved.
\end{proof}
This result establishes the fact that in the space of $N$-qubit product states, there always exist a symmetric product probe state that maximizes the QFI. Consequently, from Theorem~\ref{th1}, we can conclude the following.
\begin{corollary}
    Maximal QFI in the case of ITDM with product probe states is achieved by symmetric states of the form $\ket{\Psi^S}=\ket{\psi_c}\ket{\psi}^{\otimes N}$ where we have to optimize over only the doubly-parametrized set of states, $\ket{\psi}=\sqrt{p_s}\ket{0} + e^{i\theta_s}\sqrt{1-p_s}\ket{1}$ with $\theta_s\in[0,2\pi]$ and $p_s\in[0,1]$ along with optimal control qubit and encoding axis.
    \label{cor1}
\end{corollary}

\section{Estimating  phases with ITDM}
\label{sec:phase_estimation_itdm}
After establishing the advantage of TF operation in the context of quantum metrology, let us concentrate on an explicit phase estimation scheme and exhibit how one can achieve Heisenberg limit with TF operation. In particular, given a unitary $U_\theta=e^{-i\frac{\theta}{2}\vec\sigma.\hat n},$ our goal is to estimate the value of the  parameter $\theta$, where $\hat n = \hat xn_1 + \hat yn_2 + \hat zn_3$ is referred as the encoding axis. As realized in the previous section, by initializing the input state as $\ket{\Psi^S}=\ket{\psi_c}\ket{\psi}^{\otimes N},$  we find QFI in terms of
\begin{eqnarray}
    \nonumber\alpha &=& in_3(p_s - 1/2) \\\nonumber&&+ i\sqrt{p_s(1-p_s)}(n_1\cos{\theta_s}+n_2\sin{\theta_s}), \notag\\\nonumber
\beta &=& in_3(p_s - 1/2)  \\\nonumber &&+ i\sqrt{p_s(1-p_s)}(n_1\cos{\theta_s}-n_2\sin{\theta_s}), \notag\\
|\dot{\chi}| &=& 1/4, \,\text{and}\,\, |\dot{\omega}| = 1/4,
\end{eqnarray}
in Eq.~\eqref{eq:QFI_singleparam_Unitary_SLD0}. In particular, $\mathcal{Q}_\theta(A,B)=AN^2+BN,$ where
\begin{eqnarray}
    A=4p_c(1-p_c)*4p_s(1-p_s)n_2^2\sin^2{\theta_s},
    \label{eqn:phase_estimation_A}
\end{eqnarray}
and 
\begin{eqnarray}
   \nonumber B&=&4p_c\Bigg[\frac{1}{4} - \Big\{n_3(p_s - 1/2) \nonumber\\
&&+ \sqrt{p_s(1-p_s)}(n_1\cos{\theta_s}+n_2\sin{\theta_s})\Big\}^2\Bigg] \nonumber\\
&&+ \left. 4(1-p_c)\Bigg[\frac{1}{4} - \Big\{n_3(p_s - 1/2) \right.\nonumber\\ \nonumber
&&+ \sqrt{p_s(1-p_s)}(n_1\cos{\theta_s}-n_2\sin{\theta_s})\Big\}^2\Bigg],\\
\end{eqnarray}
is independent of the parameter $\theta$ itself which we want to estimate, and also on $\theta_c$. Notice that NAC, in this case reduces to  $(\vec\sigma.\hat n)^T=\vec\sigma.\hat n$ which implies $\hat n_2=0$. 

Let us now find the maximum value of QFI, $\bar{\mathcal{Q}}_\theta^{\hat n}$ for a given encoding direction $\hat n$,  which requires optimization over the control and the probe state. To address this, let us group the parameters $\theta_s,p_s,p_c$ in a set $\bm\mu$, i.e., $\bm\mu=\{\theta_s,p_s,p_c\}$. Mathematically, the optimization can be cast into the form  $\mathcal{Q}_\theta^{\hat n}=\max_{\bm\mu}\mathcal{Q}_\theta(A,B)$. Since optimization over multiple (three) parameters  is not straightforward, we maximize the coefficient $A$ to obtain $M_A$. Looking at the form of $A$ in Eq.~\eqref{eqn:phase_estimation_A}, we can easily find that the maximum value of  $A=n_2^2$ with $\bm\mu_A=\{\frac{\pi}{2},\frac{1}{2},\frac{1}{2}\}$ which leads to QFI $\mathcal{Q}_\theta|_{\bm\mu_A}=N^2n_2^2 + N(1-n_2^2)$. Let us now prove in the following Theorem that calculating maximum advantage of time flip, i.e., $M_A$, maximizes the QFI.
% Interestingly, in Theorem~\ref{th:opt_product}, we show that the maximum value of quantum  Fisher information for a fixed encoding axis $\hat n$ is
% $\bar{\mathcal{Q}}_\theta^{\hat n}=\mathcal{Q}_\theta|_{\bm\mu_A}=N^2n_2^2 + N(1-n_2^2)$.
\begin{theorem}
    The maximum value of QFI, achieved by a $N$-qubit symmetric product probe state in the ITDM protocol for phase estimation is $\bar{\mathcal{Q}}_\theta^{\hat n}=N^2n_2^2 + N(1-n_2^2)$.
    \label{th:opt_product}
\end{theorem}
\begin{proof}
    To demonstrate that \( \mathcal{Q}_\theta|_{\bm\mu_A} \) represents the optimal value of the QFI, we have to prove $\mathcal{Q}_\theta|_{\bm\mu_A}-\mathcal{Q}_\theta(A,B)\geq 0$, for all other values of $\theta_s, p_s, p_c$, i.e., 
    \begin{eqnarray}
       \nonumber &&\quad \quad \hspace{0.2cm}N^2n_2^2+N(1-n_2^2)-(AN^2+BN)\geq 0,\\
        &&\implies N(n_2^2-A)\geq B+n_2^2-1.
        \label{eq:ab1}
    \end{eqnarray}
    As $n_2^2$ is the optimal value of $A$, $n_2^2-A\geq 0$. Therefore, if Eq.~\eqref{eq:ab1} is true for $N=1,$ then it remains valid for all values of $N>1$. For $N=1$, it simplifies to  $A+B\leq 1$, which can be proven (see Appendix.~\ref{app:a+b}). Hence, $\bar{\mathcal{Q}}_\theta^{\hat n}=N^2n_2^2 + N(1-n_2^2)$.
\end{proof}
Let us now determine the optimal encoding direction $\hat n$ for which the maximum of QFI is attained. To do so, we define
\begin{eqnarray}
    \bar{\mathcal{Q}}_\theta=\max_{\hat{n}}\bar{\mathcal{Q}}_\theta^{\hat n}=\max_{\hat{n},\bm\mu}\mathcal{Q}_\theta(A,B),
\end{eqnarray}
which includes optimization over $\hat n$, as well as input and control qubits. As a direct consequence of Theorem~\ref{th:opt_product}, the following corollary can be established by choosing $n_2 = 1$. 
\begin{corollary}
    By fixing the optimal encoding axis $\hat n=\hat y$ and $\bm\mu_A=\{\frac{\pi}{2},\frac{1}{2},\frac{1}{2}\}$, the maximum  QFI turns out to be $N^2$, achieving Heisenberg scaling.
\end{corollary}
\textbf{Remark.} The maximum QFI for the estimation of the axis parameter of \( U = \exp\left(-i\frac{\theta}{2} \vec{\sigma} \cdot \hat{n}\right) \), where either \(\phi\) or \(\xi\) of the axis \(\hat{n} = \hat{x} \sin{\phi} \cos{\xi} + \hat{y} \sin{\phi} \sin{\xi} + \hat{z} \cos{\phi}\) is to be estimated, is determined to be \( 4N^2 \) following Eq.~(\ref{eq:QFI_singleparam_Unitary_SLD0}) (for details, see Appendix.~\ref{app:axis}). This explicitly demonstrates that our protocol based on TF is capable of achieving the Heisenberg limit for arbitrary parameters in unitaries, as expected from our derivations in Section \ref{sec:ITDM_framework}.  For example, it can be used to estimate the errors in the parameters of a gate-set \cite{Kimmel2015}, required for universal quantum computation.

\begin{comment}
\ph{Note that the maximum QFI for the estimation of the axis parameter of \( U = \exp\left(-i\frac{\theta}{2} \vec{\sigma} \cdot \hat{n}\right) \), where either \(\phi\) or \(\xi\) of the axis \(\hat{n} = \hat{x} \sin{\phi} \cos{\xi} + \hat{y} \sin{\phi} \sin{\xi} + \hat{z} \cos{\phi}\) is to be estimated, is determined to be \( 4N^2 \) following Eq.~(\ref{eq:QFI_singleparam_Unitary_SLD0}). This demonstrates that our protocol is capable of achieving the Heisenberg limit for arbitrary unitary parameters, as detailed in Appendix.~\ref{app:axis} }.
\end{comment}

{\it Measurement scheme to attain maximum QFI.} Let us exhibit that the maximum QFI, i.e., $\mathcal{Q}_\theta=N^2$ for optimal encoding is, in fact, achievable by performing local measurement during decoding. The scheme consists of the following steps: $(1)$ Prepare the control state as $\ket{\psi_c}=\ket{+}=(\ket 0+\ket 1)/\sqrt{2}$ and each qubit of the probe state as $\ket \psi=\ket{+i}=(\ket 0+i\ket 1)/\sqrt{2}$, i.e., the input state is $\ket{\Psi^s}=\ket{+}\ket{+i}^{\otimes N}$. $(2)$ The encoded state after the application of $\mathbb T_{U_\theta}$ with $U_\theta=e^{-i\theta\sigma_y/2}$ (i.e., $\hat n=\hat y$) on $\ket{\Psi_s}$ is given by
\begin{eqnarray}
   \nonumber \ket \Phi = \frac{1}{\sqrt 2}e^{-iN\theta/2}\left(\ket{0}+e^{iN\theta}\ket 1 \right)\otimes\ket{+i}^{\otimes N}.\\
   \label{eq:phi_i_y}
\end{eqnarray}
$(3)$ Perform projective measurement on the control qubit in the basis of $\sigma_x$, i.e., $\{\ket{+},\ket{-}\}$. $(4)$ Repeat the above procedure $\lambda$ number of times and record the measurement outcomes.

The probability statistics for the outcome $\ket\pm$ is given by $p(\pm|\theta) = \frac{1}{2}[1\pm\cos{(N\theta )}]$. Therefore, the FI \cite{Braunstein1994May} coincides with $\bar{\mathcal{Q}}_\theta$, i.e.,
\begin{eqnarray}
    \mathcal{F}_\theta&=&\sum_{b\in\{+,-\}} p(b|\theta)\left[\diffp{\ln{p(b|\theta)}}{\theta}\right]^2,
    \label{eq:cfi}
    \\&=& N^2=\bar{\mathcal{Q}}_\theta.
    \label{eq:cfi_1}
\end{eqnarray} 
% \ph{The Cramer-Rao bound \cite{Braunstein1994May} \ph{(have to study - Only when the estimator is unbiased, its RMSE satisfies Cramer-Rao bound - see chiribella IDC metrology)} for root mean square error (RMSE) is achievable in large $\lambda $ limit, i.e., 
% \begin{eqnarray}
%     \Delta \theta=\frac{1}{\sqrt{\lambda \mathcal{F}_\theta}}=\frac{1}{\sqrt\lambda N}.
% \end{eqnarray}}
Notably, in this scenario, the probe and control states (Eq.~\eqref{eq:phi_i_y}) remain as a pure product state even after bidriectional encoding along $y$-axis, which is a global operation, while still providing HL. This demonstrates that entanglement is not necessary for surpassing SQL in ITDM protocol. However, the two qubit unitaries generated by time flip may generate entanglement for arbitrary input product probe states.

The above-discussed protocol is optimal when $\hat n=\hat y$ and the root mean square error  scales as $1/N$. However, it is interesting to ask -- \textit{Given an arbitrary direction of the encoding axis, how does FI, $\mathcal F_\theta$, scale when the decoding measurement basis is fixed in $\{\ket{+}, \ket -\}?$} We answer this question in the following subsection.

\subsection{Fisher information by restricted measurements}
\label{sec:control_qubit_measure}

By initializing the control and probe states as a product state (Eq.~\eqref{eq:PureStateOutput}) and encoding the axis as $\hat n$, let us consider the decoding measurement only on the control qubit in $\{\ket +, \ket -\}$ basis. In this scenario, the outcome statistics read
\begin{align}
    \nonumber &p(\pm|\theta) = \Tr\left[(\ket \pm\bra{\pm}_c \otimes I_p) \ket{\Phi}\bra \Phi\right]\\ \nonumber &= \frac{1}{2} \pm \frac{\sqrt{p_c(1 - p_c)}}{2} \bigg( e^{i \theta_c} \prod_{i=1}^N \langle \chi_i | \omega_i \rangle + e^{-i \theta_c} \prod_{i=1}^N \langle \omega_i | \chi_i \rangle \bigg).\\
\end{align}
Since $\braket{\chi_i}{\omega_i}$ is a complex number, we parametrize it as $\braket{\chi_i}{\omega_i}=r_i({\bm\nu})e^{if_i({\bm\nu})}$ where ${\bm\nu}=\{\theta,\theta_s,p_s,\hat n\}$. According to this parametrization, 
\begin{align}
    \nonumber &p(\pm|\theta)\\\nonumber &=\frac{1}{2} \pm \sqrt{p_c (1 - p_c)}  \cos \left( \sum_{i=1}^{N} f_i(\bm\nu) + \theta_c \right)\prod_{i = 1}^{N}r_i(\bm\nu),\\
\end{align}
and the FI can be calculated as
\begin{widetext}
% \ph{Therefore, the derivatives of $p(\pm|\theta)$ is given by
% \begin{eqnarray}
%     \nonumber \diffp{p(\pm|\theta)}{\theta} &=& \pm\sum_{j = 1}^{N}\sqrt{p_c (1 - p_c)} \dot{r_j}(\bm\nu) \left(\prod_{i \neq j}^{N}r_i(\bm\nu)\right) \cos \left( \sum_{i=1}^{N} f_i(\bm\nu) + \theta_c \right) \nonumber \\ \nonumber
% && \mp \sqrt{p_c (1 - p_c)} \left(\prod_{i=1}^{N}r_i(\bm\nu)\right) \sin \left( \sum_{i=1}^{N} f_i(\bm\nu) + \theta_c \right)\left(\sum_{i=1}^{N}\dot{f_i}(\bm\nu)\right).\\ \label{eq:DerProbplus_singleparame_controlqubitmeasure}
% \end{eqnarray}}
\begin{align}
        \nonumber &\mathcal F_{\theta}&\\\nonumber &= 4 p_c (1 - p_c) \bigg[\sum_{j = 1}^{N}\bigg(\prod_{i \neq j}^{N}r_i(\bm\nu)\bigg)\bigg\{\dot{r_j}(\bm\nu) \cos \left( \sum_{i=1}^{N} f_i(\bm\nu) + \theta_c \right) - \left(\frac{r_j(\bm\nu)}{N}\right) \sin \left( \sum_{i=1}^{N} f_i(\bm\nu) + \theta_c \right)\left(\sum_{i=1}^{N}\dot{f_i}(\bm\nu)\right) \bigg\} \bigg]^2 \\ &\quad\times\left[1 - 4p_c (1 - p_c)  \left(\prod_{i=1}^{N}r_i(\bm\nu)\right)^2 \cos^2 \left(\sum_{i=1}^{N}{f_i}(\bm\nu) + \theta_c \right)\right]^{-1},
\label{eq:FI_singleparam_contolqubitmeasure}
\end{align}
which can be shown to be optimal at $p_c=1/2$. We consider $\ket{\psi_i} = \ket{\psi}$ for all $i$, which is determined to be optimal based on rigorous numerical investigations. This implies $r_i(\bm\nu)=r(\bm\nu)$, $f_i(\bm\nu)=f(\bm\nu)$ $\forall i$ and along with $p_c=1/2$, we have 

    \begin{eqnarray}
        \mathcal{F}_\theta=\frac{N^2  r^{2N-2}(\bm\nu) \left\{ \dot{r}(\bm\nu) \cos \left( N f(\bm\nu) + \theta_c \right) - r(\bm\nu) \dot{f}(\bm\nu) \sin \left( N f(\bm\nu) + \theta_c \right) \right\}^2 }{1 -  r^{2N}(\bm\nu) \cos^2 \left( N f(\bm\nu) + \theta_c \right)}.
\label{eq:FI_singleparam_contolqubitmeasure}
    \end{eqnarray}
\end{widetext}
To approach the question of the optimal probe state, similar to QFI, we define
\begin{equation}
    \bar{\mathcal{F}}_\theta = \max_{\hat n}\bar{\mathcal{F}}_\theta^{\hat n},
\end{equation} 
where $\bar{\mathcal{F}}_\theta^{\hat n}$ is the maximum  FI for a given axis $\hat n$ by choosing optimal probe and performing decoding measurement in $\{\ket +,\ket -\}$ basis on optimal control qubit, i.e., $\bar{\mathcal{F}}_\theta^{\hat n} = \max_{\bm\mu,\theta_c}\mathcal{F}_\theta$. Note that $U_\theta^\dagger U_\theta^T$ being unitary, we have $0\leq r(\bm\nu)\leq 1$. Therefore, from Eq.~\eqref{eq:FI_singleparam_contolqubitmeasure}, it is evident that if $r(\bm\nu)<1$,  $\mathcal F_\theta$ decreases exponentially with increasing value of $N$. This means that it is desirable to set $r(\bm\nu)=1$ by taking initial probe state $\ket \psi=\ket{u(\hat n,\theta)^{\pm}}$ where $\ket{u(\hat n,\theta)^{\pm}}=(\mathrm g^\mp(\hat n, \theta)\ket{0}+\ket{1})/\sqrt{1+|\mathrm g^\mp(\hat n, \theta)|^2}$ is the eigenstate of $U_\theta^\dagger U_\theta^T$, i.e., $U_\theta^\dagger U_\theta^T\ket{u(\hat n,\theta)^{\pm}} = \exp\left({\pm i \mathrm{f}(n_2, \theta)}\right)\ket{u(\hat n,\theta)^{\pm}}$.
In this scenario, $p_s= \mathrm g^\mp(\hat n,\theta)^2/\sqrt{1+|\mathrm g^\mp(\hat n, \theta)|^2}$ and $\theta_s=i \ln(1+|\mathrm g^\mp(\hat n, \theta)|^2 - \mathrm g^\mp(\hat n, \theta)^2)/2$. Moreover, we have 
\begin{eqnarray}
    \nonumber\mathrm f(n_2,\theta)=\tan^{-1}\left(\frac{\sqrt{4n_2^2\sin^2{\frac{\theta}{2}}\left(1-n_2^2 \sin^2{\frac{\theta}{2}}\right)}}{1-2n_2^2\sin^2{\frac{\theta}{2}}}\right),\\
\end{eqnarray} 
and the explicit forms of $g^\mp(\hat n,\theta)$ are given in Appendix.~\ref{app:udut}, which shows that the eigenvectors depends on the encoding axis and the parameter $\theta$ itself. This implies $\dot r(\bm\nu)=0$ and  $f(\bm\nu)=\mathrm f(n_2, \theta)$. Therefore, preparing the probe $\ket{\psi}^{\otimes N}$ in either $\ket{u(\hat n, \theta)^+}^{\otimes N}$ or $\ket{u(\hat n, \theta)^-}^{\otimes N}$ and $\ket{\psi_c}=(\ket 0+e^{i\theta_c}\ket 1)/\sqrt{2}$ and measuring the control qubit in the $\{\ket +,\ket -\}$ basis lead to FI (following Eq.~\eqref{eq:FI_singleparam_contolqubitmeasure})
 \begin{eqnarray}
     \bar{\mathcal F}_\theta^{\hat n} = N^2\dot{\mathrm f}(n_2, \theta)^2.
 \end{eqnarray}
However, the complexity is two-folded -- $(1)$ the input state parameters, $\theta_s,p_s,$ is $\theta$-dependent, $(2)$ $\bar{\mathcal F}_\theta^{\hat n}$ is $\theta$-dependent which means that the estimation accuracy may fluctuate or degrade depending on the true value of $\theta$. Notice that choosing $\hat n=\pm \hat y$ gets rid of such complexities and gives $\mathrm{f}=-\theta\pm \pi$ which achieves $\bar{\mathcal F}_\theta = N^2=\mathcal{\bar Q}_\theta$ which is consistent with Eq.~\eqref{eq:cfi_1}. Next, we show that the above-mentioned subtlety can be removed when the parameter $\theta$ is taken to be small.

% \ph{Numerically check whether symmetric probe state is optimal for here also???...}. \ga{yes with a probability of 0.9}

\subsubsection{Optimal input probes for small $\theta$}

Let us write the input probe as $\rho=\ketbra{\psi}{\psi}=(I+\hat{s}.\vec \sigma)/2$ where $\hat{s}=\hat xs_x + \hat ys_y + \hat zs_z$ and $s_x=2\sqrt{p_s(1-p_s)}\cos{\theta_s}, ~s_y=2\sqrt{p_s(1-p_s)}\sin{\theta_s},~ \text{and } s_z=(2p_s-1)$. The overlap $\braket{\chi}{\omega}$ can be written as
\begin{eqnarray}
    \nonumber\braket{\chi}{\omega} &=& \langle\psi|U_\theta^\dagger U_\theta^T|\psi\rangle = \Tr(U_\theta^T \rho U_\theta^\dagger) \\ &=& 1 - n_2^2(1 - \cos\theta) + in_2n_3s_x(1-\cos\theta)\nonumber\\
&&  -in_1n_2s_z(1-\cos\theta) +in_2s_y\sin\theta.
\label{eq:xw_smallt}
\end{eqnarray}
Therefore, the modulus and phase of the complex number in Eq.~\eqref{eq:xw_smallt} is given by
\begin{align}
   \nonumber  r(\bm\nu) &= \bigg[\left\{1 - n_2^2(1 - \cos\theta)\right\}^2 +\big\{n_2s_y\sin\theta \\\nonumber &\quad+(n_2n_3s_x-n_1n_2s_z)(1-\cos\theta)\big\}^2\bigg]^{1/2},\\
   \label{eq:r_nu}
\end{align}
and 
\begin{align}
    \nonumber &f(\bm\nu) \\\nonumber &= \tan^{-1}\frac{n_2s_y\sin\theta+(n_2n_3s_x-n_1n_2s_z)(1-\cos\theta)}{1 - n_2^2(1 - \cos\theta)},\\
    \label{eq:f_nu}
\end{align}
respectively. When $\theta$ is very small ($\theta=\tilde \theta$), expanding Eqs.~\eqref{eq:r_nu} and \eqref{eq:f_nu} upto second order in $\theta$ we obtain (see details in  Appendix.~\ref{app:opt_probe_control})
\begin{eqnarray}
    \nonumber \mathcal{F}_{\tilde\theta}&=&s_y^2n_2^2N^2 + (1-s_y^2)n_2^2 N\delta_{\theta_c,0},\\ &=& s_y^2n_2^2N(N-\delta_{\theta_c,0})+n_2^2N \delta_{\theta_c,0},
    \label{eq:ftt}
\end{eqnarray}
where $\delta_{\theta_c,0}=1$ if $\theta_c=0$ and $\delta_{\theta_c,0}=0$ elsewhere. Interestingly, in the limit $N\to\infty$, Fisher information  $\mathcal{F}_{\tilde\theta}\big|_{N\to\infty}=s_y^2n_2^2N^2$ is $\tilde\theta$-independent and reaches HL irrespective of any probe state. From Eq.~\eqref{eq:ftt}, it immediately follows that the optimal probe for  $\tilde\theta$ estimation is given by $s_y=\pm 1$, i.e., $\bm\mu=\bm\mu_A=\{\pi/2+2m\pi,1/2,1/2\}\forall m\in\mathbb{Z}$  which implies 
\begin{eqnarray}
    \bar{\mathcal{F}}^{\hat n}_{\tilde\theta}=n_2^2 N^2 = \mathcal{\bar Q}_{\tilde\theta}^{\hat n}|_{N\to\infty}.
\end{eqnarray}
Hence, we arrive at the following result, which shows that the two-fold complexity previously described can be eradicated under the given scenario: 
\begin{theorem}
    The FI, $\bar{\mathcal{F}}^{\hat{n}}_{\tilde\theta} = n_2^2 N^2$, independent of $\tilde \theta$ to be estimated which is small, saturates the QFI, $\mathcal{\bar{Q}}_{\tilde\theta}^{\hat{n}} \big|_{N \to \infty}$ asymptotically, using the control qubit $\ket{\psi_c} = (\ket{0} + e^{i\theta_c} \ket{1}) / \sqrt{2}$ along with the probe state as $\ket{\pm i}=(\ket{0}+i\ket{1})/\sqrt{2}$ and the decoding measurement in the $\sigma_x$ basis.
    % In the limit $N \to \infty$, independent of the small value $\tilde{\theta}$ to be estimated, the classical FI, $\bar{\mathcal{F}}^{\hat{n}}_{\tilde\theta} = n_2^2 N^2$, saturates the QFI, $\mathcal{\bar{Q}}_{\tilde\theta}^{\hat{n}} \big|_{N \to \infty}$, using the control qubit $\ket{\psi_c} = (\ket{0} + e^{i\theta_c} \ket{1}) / \sqrt{2}$ along with the probe state as $\ket{\pm i}=(\ket{0}+\ket{1})/\sqrt{2}$ and decoding measurement in the $\sigma_x$ basis.
\end{theorem}
% Therefore, the FI, $\bar{\mathcal{F}}^{\hat n}_{\theta}$ corresponding to any encoding direction can accomplish HL 
% which is $\theta$-independent also does not depend on the phase $(\theta_c)$ of the control qubit. Moreover, the probe state is also $\tilde\theta$-invariant.
% Hence we arrive at the following result --
% \begin{theorem}
%     For small values of the parameter ($\theta$) that we are estimating, encoded along $\hat n$, initializing the probe in $\ket{+i}^{\otimes N}$ and control in $\ket +$ we can achieve RMSE $\Delta\theta=\frac{1}{\sqrt{\lambda}|n_2| N}$ (Heisenberg scaling) by measurement in $\{\ket +,\ket -\}$ basis. 
% \end{theorem}
% So far we have been dealing with product states as the initial probe state. In the next section, we explore the scenario of the entangled state as a probe state.

\section{Entanglement can not beat $N^2$-scaling}
\label{sec:entanglement}

Our analysis upto now has focused on utilizing product states as the initial probe state. Let us now examine whether the use of an entangled probe can improve scaling, surpassing the Heisenberg limit or not. To address the same, we write the general $N$-qubit probe state as 
\begin{eqnarray}
    \ket{\Psi} = \sum_{i=1}^{2^N} \Lambda_i\ket{\lambda_i},
    \label{eq:state_all}
\end{eqnarray}
where $\sum_{i=1}^{2^N}\ketbra{\lambda_i}{\lambda_i} = I_{2^N}$ and $\braket{\lambda_i}{\lambda_j}=\delta_{ij}$. Defining $\ket{\lambda_i}=\bigotimes_{j=1}^N\ket{\psi^i_{j}},$ we rewrite $\ket{\Psi} = \sum_{i=1}^{2^N} \Lambda_i \bigotimes_{j=1}^N\ket{\psi^i_{j}},$ where $\ket{\psi_j^i}$ is the $i$th basis state for the $j$th qubit with $\braket{\psi_j^i}{\psi_j^k}=\delta_{ik}$ and $\sum_{i=1}^2\ketbra{\psi_j^i}{\psi_j^i}=I_2$. Extending beyond the use of product states as probes, our theorem below indicates that entanglement does not provide any additional advantage in achieving scaling of the QFI with the number of qubits beyond the HL.
\begin{theorem} 
In the ITDM protocol, the maximum QFI extractible from any probe state, including all entangled states, is upper bounded by $\gamma N^2 + \zeta N,$ where \( \gamma \) and \( \zeta \) are constants, independent of \( N \).

    \label{th_ent}
\end{theorem}
\begin{proof}
Let us present a brief outline of our comprehensive proof, (for details, see Appendix.~\ref{app:th2}). Using the general $N$-qubit  probe state in Eq.~\eqref{eq:state_all}, we derive an upper bound on QFI as
\begin{align}
    \mathcal{Q}_{\theta} &\leq {\tilde{\mathcal{Q}}_\theta}= 4\langle \dot{\Phi} | \dot{\Phi} \rangle, \nonumber\\  &=4\sum_{i=1}^{2^N}\sum_{k=1}^{2^N}\Lambda_i\Lambda_k^* \Bigg[ p_c \sum_{m=1}^{N} \Big( \dot\chi_{m}^{ki} \prod_{j \neq m} a_{j}^{ki} \nonumber\\ 
& \qquad \qquad + \sum_{l \neq m} \alpha_{m}^{ki} \alpha_{l}^{ik *} \prod_{j \neq m, l} a_{j}^{ki} \Big) \nonumber\\
&\quad + (1 - p_c)\sum_{m=1}^{N} \Big( \dot\omega_{m}^{ki} \prod_{j \neq m} b_{j}^{ki} \nonumber\\
& \qquad \qquad + \sum_{l \neq m} \beta_{m}^{ki} \beta_{l}^{ik *} \prod_{j \neq m, l} b_{j}^{ki} \Big) \Bigg],
\label{eq:uppbound}
\end{align}
\begin{figure*}
\begin{subfigure}
\centering
\includegraphics[width=0.3\textwidth]{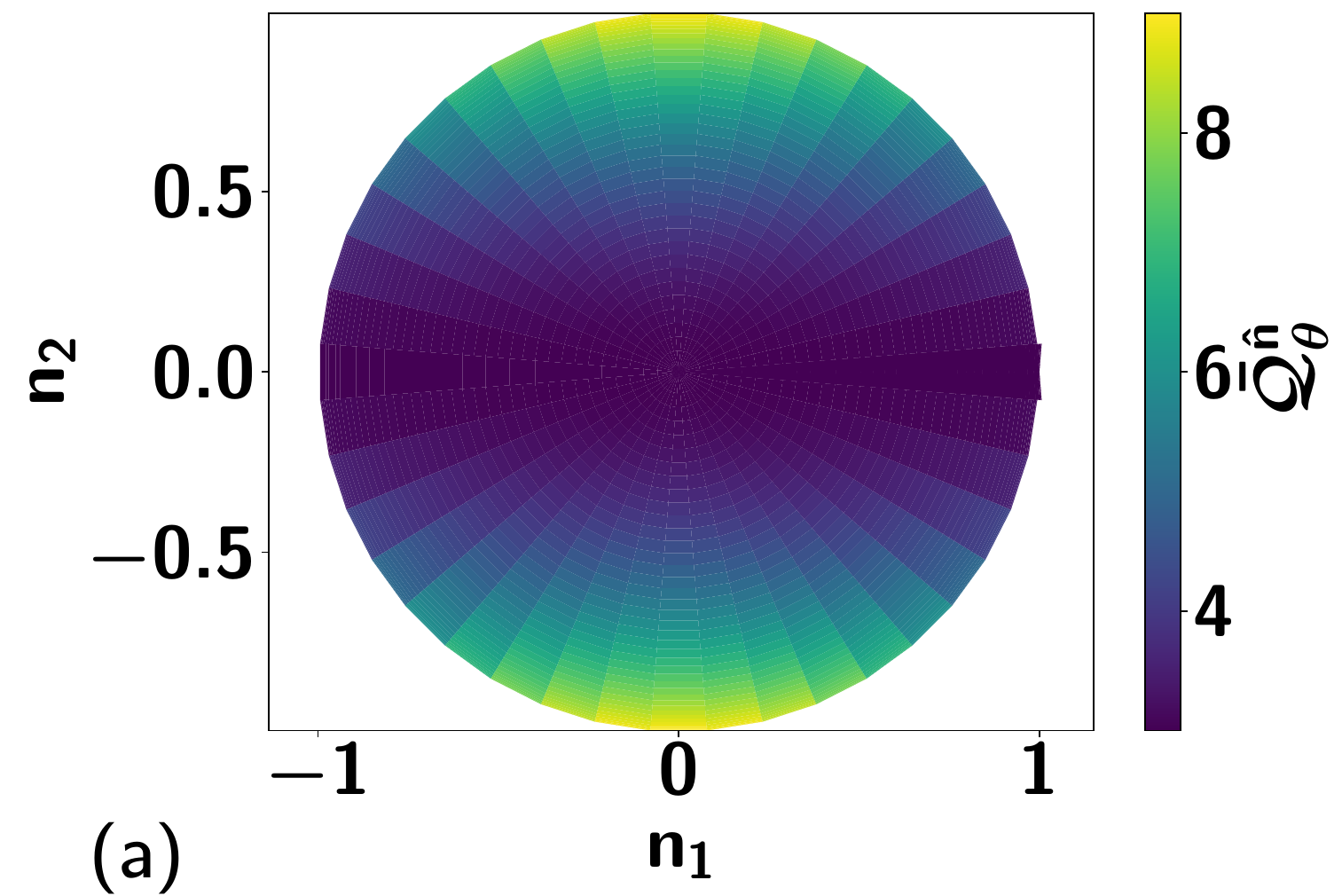}
\end{subfigure} 
\begin{subfigure}
\centering
\includegraphics[width=0.3\textwidth]{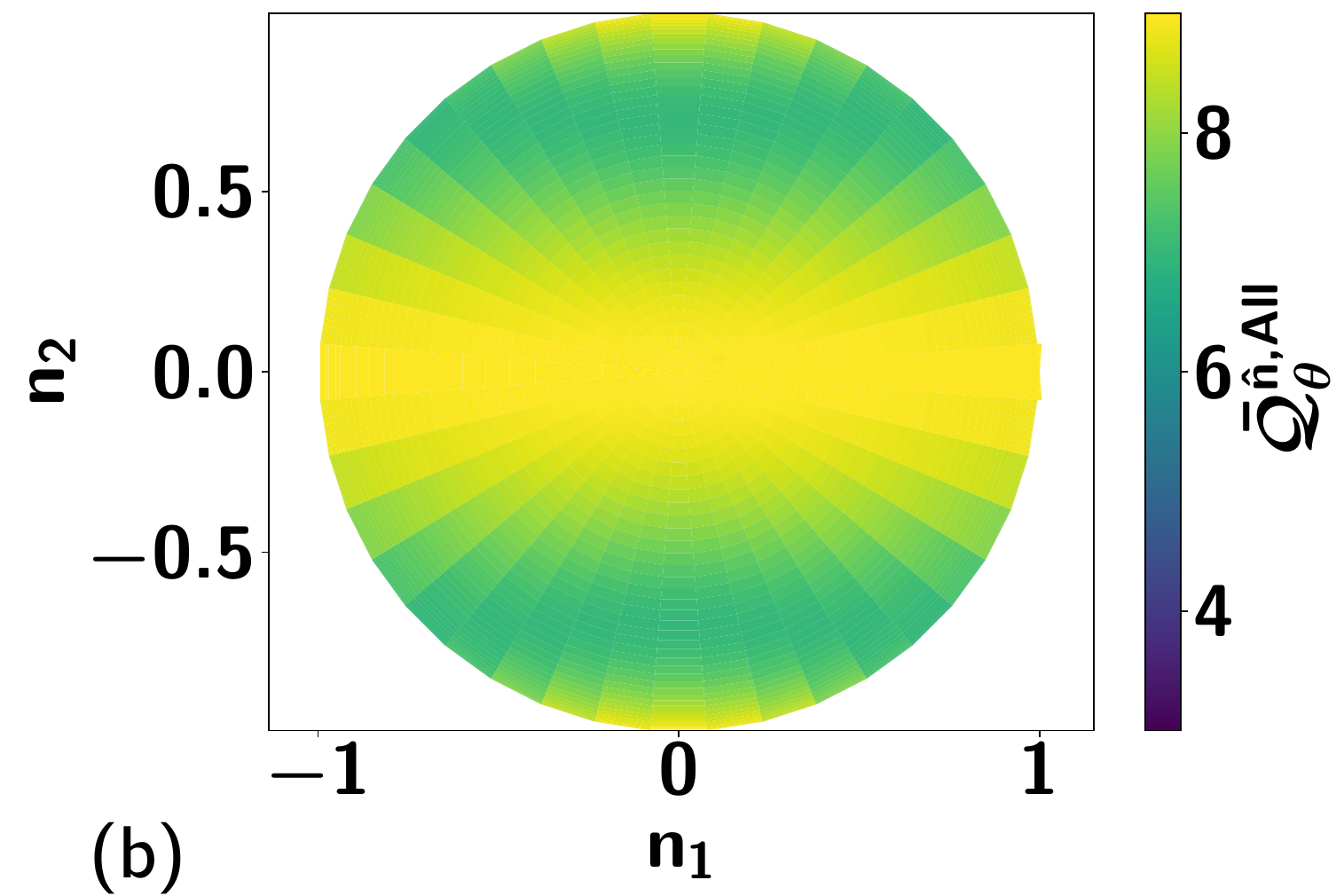}
\end{subfigure} 
\begin{subfigure}
\centering
\includegraphics[width=0.3\textwidth]{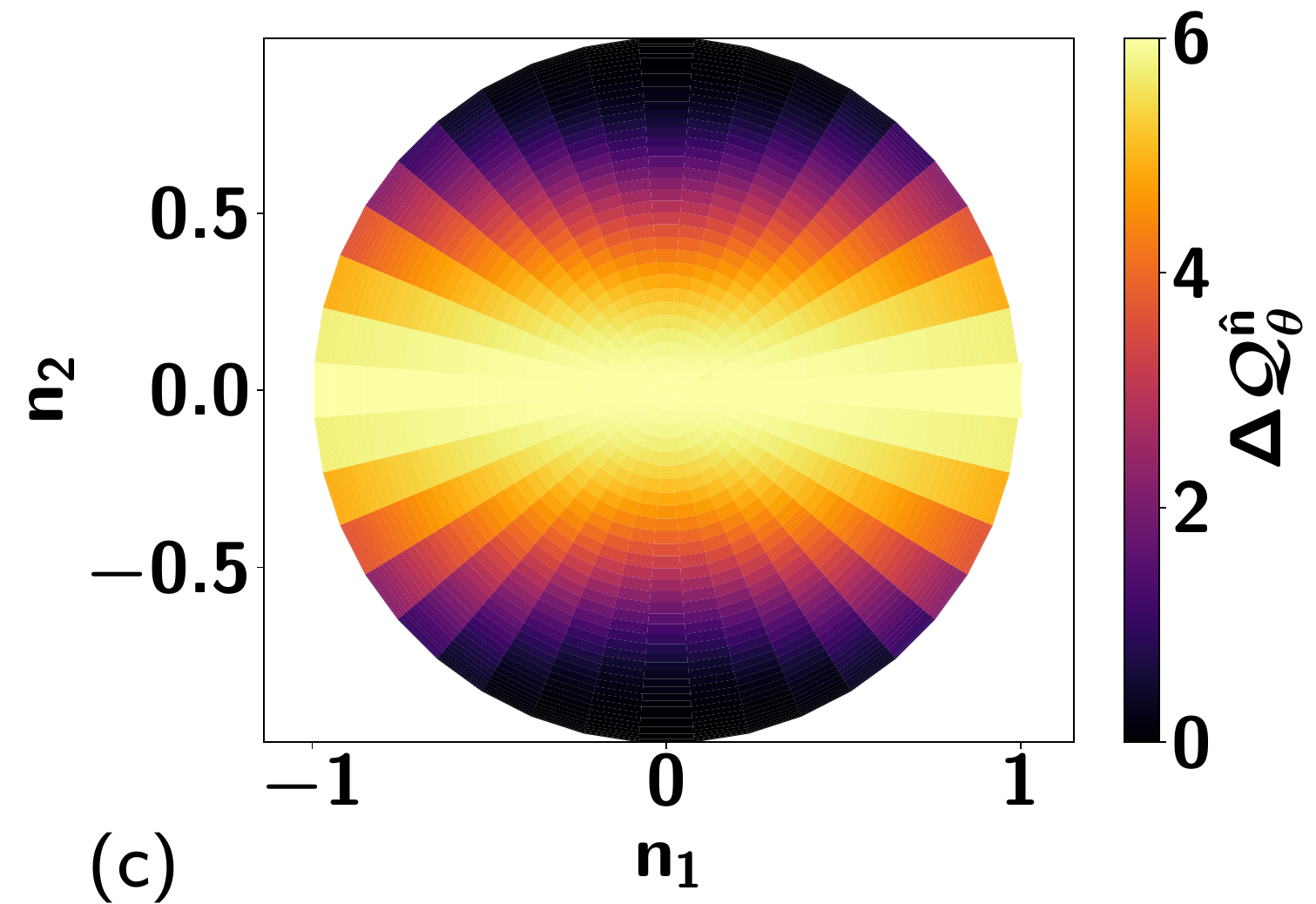}
\end{subfigure} 
\caption{ \textbf{QFI  using three-qubit probe states for  estimating \(\theta\) in \(U_{\theta}\).}  Map plot of (a) $\mathcal{\bar Q}_\theta^{\hat n}$, optimized over all product states, (b) $\mathcal{\bar Q}_\theta^{\hat n,All}$ optimized over all states,   (c) $\Delta \mathcal Q_\theta^{\hat n}=\mathcal{\bar Q}_\theta^{\hat n,All}-\mathcal{\bar Q}_\theta^{\hat n}$ against encoding axis, \(n_1\) (horizontal axis) and \(n_2\) (vertical axis). We see that the difference vanishes  near $n_2 = \pm 1$ precisely where both (a) and (b) attain their maximal QFI of $N^2=9$. We conclude that product states can give as much QFI as entangled states when optimized over the axis for a three-qubit probe system. All the axis are dimensionless.}
\label{fig:optimalstate_vs_ourstate}
\end{figure*}
where 
\begin{eqnarray}
\langle \dot\chi^k_{m} | \dot\chi^i_{m} \rangle &=& \dot\chi_{m}^{ki}, \quad \langle \chi^k_{j} | \chi^i_{j} \rangle = a_{j}^{ki}, \nonumber\\
\langle \dot\omega^k_{m} | \dot\omega^i_{m} \rangle &=& \dot\omega_{m}^{ki}, \quad \langle \omega^k_{j} | \omega^i_{j} \rangle = b_{j}^{ki}, \nonumber\\
\langle \dot\chi^k_{m} | \chi^i_{m} \rangle &=& \alpha^{ki}_{m}, \quad \langle \dot\omega^k_{j} | \omega_{j}^i \rangle = \beta^{ki}_{j}, 
\end{eqnarray} and $\ket{\chi_j^i} = U_\theta\ket{\psi^i_j}$ and $\ket{\omega_j^i} = U_\theta^T\ket{\psi^i_j}$. In Eq.~\eqref{eq:uppbound}, the terms $\prod_{j \neq m} a_{j}^{ki} $, and $\prod_{j \neq m} b_{j}^{ki} $ are the products of the inner products of $(N-1)$ qubits. Hence, if the qubits in $\ket{\lambda_i}$ and $\bra{\lambda_k}$ differ in more than one position, $\prod_{j \neq m} a_{j}^{ki} = \prod_{j \neq m} b_{j}^{ki}=0$. Similarly, if the qubits in $\ket{\lambda_i}$ and $\bra{\lambda_k}$ differ in more than two positions, $\prod_{j \neq m, l} b_{j}^{ki}=\prod_{j \neq m, l} a_{j}^{ki}=0$. Defining, $\mathcal{\tilde{Q}}_{\theta}^{ X}$ as the term which involves contribution only from $N$-qubit basis states where qubits differ in $X$ number of positions, we can write $\mathcal{\tilde{Q}}_{\theta}=\sum_{X=0}^N\mathcal{\tilde{Q}}_{\theta}^{ X}=\sum_{X=0}^2\mathcal{\tilde{Q}}_{\theta}^{ X}$ since, $\mathcal{\tilde{Q}}_{\theta}^{ X}=0$ $\forall X\geq 3$. In Appendix.~\ref{app:th2}, we show that each of these three terms is bounded by an expression of the form $\gamma_iN^2+\zeta_i N$. Hence, $\mathcal{\tilde Q}_\theta\leq \gamma N^2+\zeta N$ where $\gamma,\zeta$ are constants, which are independent of $N$.
\end{proof}

Although entangled states do not offer superior scaling compared to product probe states, we investigate whether they can yield higher QFI. To this aim, given an encoding axis $\hat n$, we perform numerical maximization of QFI over the whole state space (see Eq.~\eqref{eq:state_all}) for $N=3$ denoted as $\mathcal{\bar Q}_\theta^{\hat n,All}$ (see Fig.~\ref{fig:optimalstate_vs_ourstate}(b)). Interestingly, for $n_2=1$, $\mathcal{\bar Q}_\theta^{\hat y,All}$ is the same as that with QFI for optimal product probe state, i.e.,  $\mathcal{\bar Q}_\theta^{\hat y,All}=\bar{\mathcal{Q}}_\theta^{\hat y}=9=\mathcal{\bar Q}_\theta$. Notice that $\hat n=\hat y$ is indeed the universal optimal encoding axis which means maximal attainable value $(N^=9)$ of QFI by optimizing over the whole state space and the encoding axis. However, as value of $ n_2$ deviates from unity, the difference between the maximum value of QFI with optimal product probe and $\mathcal{\bar Q}_\theta^{\hat n,All}$, i.e., $\Delta \mathcal Q_\theta^{\hat n}=\mathcal{\bar Q}_\theta^{\hat n,All}-\mathcal{\bar Q}_\theta^{\hat n}$ increases significantly as demonstrated in Fig.~\ref{fig:optimalstate_vs_ourstate}(c). For $n_2=0,$ $\bar{\mathcal{Q}}_\theta^{\hat n, All}=9$ and  $\bar{\mathcal{Q}}_\theta^{\hat n}=3$ where the difference $\Delta \mathcal Q_\theta^{\hat n}$ reaches a maximum value of $6$. This observation ensures that the optimal probe states which give higher QFI than fully separable probe states must be entangled in at least some bipartition although no scaling benefit is provided by the entangled probes. Notice that $\Delta \mathcal Q_\theta^{\hat n}\approx\mathcal{O}(10^{-2})$ in the region $\pm 0.95\lesssim n_2\leq \pm1$ irrespective of $n_1$ and $n_3$ which shows the robustness of optimal product probe states compared to the entangled ones against nonoptimal encoding.

% Given an encoding axis $\hat n$, we assess the performance of QFI with optimal product-probe state: $\bar{\mathcal{Q}}_\theta^{\hat n}=N^2n_2^2+N(1-n_2^2)$ according to Theorem~\ref{th:opt_product} and optimal entangled-probe state: $\bar{\mathcal{Q}}_\theta^{\hat n, Ent}$ for $N=3$ (see Fig.~\ref{fig:optimalstate_vs_ourstate}). For encoding axis $\hat n = \pm\hat y$, both are same, i.e., $\bar{\mathcal{Q}}_\theta^{\hat n} = \bar{\mathcal{Q}}_\theta^{\hat n, Ent}=9$, which is also the optimal value over all axes ($\bar{\mathcal{Q}}_\theta=9$).  However, as value of $ n_2$ deviates from $1$, maximum value of QFI with optimal product probe and optimal entangled probe differs significantly as demonstrated in Fig.~\ref{fig:optimalstate_vs_ourstate}. For $n_2=0$ we have QFI $\bar{\mathcal{Q}}_\theta^{\hat n, Ent}=9$, $\bar{\mathcal{Q}}_\theta^{\hat n}=3$ where the difference reaches maximum value of $6$. Notice that the difference is $\mathcal{O}(10^{-2})$ in the region $\pm 0.95\lesssim n_2\leq \pm1$ irrespective of $n_1$ and $n_3$ which shows the robustness of optimal product-probe states compared to entangled ones for nonoptimal encoding. 

\section{Advantage in noisy unitary encoding}
\label{sec:noisy_adv}

Both theoretical and experimental perspectives of quantum metrology demand the investigation on how the protocol gets affected in the presence of noisy environment. Moreover, it was shown that the precision in estimating parameters in a noisy scenario can also be improved by using indefinite causal order \cite{chapeau2021noisy}. Therefore, it is interesting to analyze the ITDM scheme under noise. For inspection, we consider a noisy input probe, $\rho=(I+\vec s.\vec\sigma)/2$ with $|\vec s|\leq 1,$ which is influenced by depolarizing noise, $\mathcal{N}(\rho)=q\rho+(1-q)I/2$ \cite{Nielsen2012} after unitary encoding $U_\theta$. Consequently, the effective channel, $\mathcal{N}_\theta(\rho)=q U_\theta \rho U_\theta^\dagger + (1-q)I/2$ is termed as noisy unitary channel with Kraus operators $\{K_0 = \sqrt{1-p}U_\theta, K_1 = \sqrt{p/3}\sigma_xU_\theta, K_2 = \sqrt{p/3}\sigma_yU_\theta, K_3 = \sqrt{p/3}\sigma_zU_\theta\}$ where $q=1-4p/3\in[0,1]$. Application of time-flip to such a bistochastic channel \cite{chiribella2022quantum} leads to a quantum operation $T_{\mathcal N_\theta}(\rho_c \otimes \rho)=\sum_{j=0}^3 F_j (\rho_c\otimes \rho)F_j^\dagger$ with Kraus operators $F_j = \ketbra{0}{0}\otimes K_j + \ketbra{1}{1}\otimes K_j^T$ where $\rho_c=\ket{\psi_c}\bra{\psi_c}$. Therefore,  after the encoding procedure in ITDM scheme with noisy encoding operation (NITDM), the output state is given by $\Theta = \mathbb T_{\mathcal N_\theta}(\rho_c \bigotimes_{i=1}^N\rho_i) = \sum_{j \in \Pi}\mathbb{F}_j(\rho_c\bigotimes_{i=1}^N\rho_i)\mathbb{F}_j^\dagger$ with $\mathbb F_j = \ketbra{0}{0} \bigotimes_{i = 1}^N K_{j[i]} +\ketbra{1}{1}\bigotimes_{i = 1}^N{K_{j[i]}^T}$, where $\Pi$ is the set of $4^N$ integer strings $j$, with each string specifying the Kraus operator $K_{j[i]}$ acted on the qubit $i$, and $\sum_{j \in \Pi}\mathbb{F}_j^\dagger\mathbb{F}_j = I$ and $\sum_{j \in \Pi} \bigg(\bigotimes_{i=1}^NK_{j[i]}\bigg)^\dagger \bigg(\bigotimes_{i=1}^NK_{j[i]}\bigg) = I$. This can further be simplified as
\begin{align}
\Theta &= p_c |0\rangle \langle 0| \bigotimes_{i=1}^N \left( \sum_{j=0}^3 K_j \rho_i K_j^{\dagger} \right)\notag\\
&\quad+ \sqrt{p_c (1 - p_c)} e^{-i \theta_c} |0\rangle \langle 1| \bigotimes_{i=1}^N \left( \sum_{j=0}^3 K_j \rho_i {K_j^T}^{\dagger} \right) \notag\\
&\quad+ \sqrt{p_c (1 - p_c)} e^{i \theta_c} |1\rangle \langle 0| \bigotimes_{i=1}^N \left( \sum_{j=0}^3 K_j^T \rho_i K_j^{\dagger} \right) \notag\\
&\quad+ (1 - p_c) |1\rangle \langle 1| \bigotimes_{i=1}^N \left( \sum_{j=0}^3 K_j^T \rho_i {K_j^T}^{\dagger} \right).  \label{eq:NoisyProbeoutput_withChannels_allStatesame}
\end{align}
Defining $\Tr \left( \sum_j K_j^{T} \rho_i K_j^{\dagger} \right) = r_i(\bm\nu) e^{i f_i(\bm\nu)}$ with $\bm\nu=\{\theta,\vec s,\hat n,q\}$ as the updated set $\bm\nu$, outcome probabilities of measurement in the $\sigma_x$-basis on the control qubit can be written as
\begin{align}
\nonumber&p(\pm|\theta)&\\\nonumber &= \frac{1}{2} \pm \sqrt{p_c (1 - p_c)}  \cos \left( \sum_{i=1}^N f_i(\bm\nu) + \theta_c \right)\prod_{i=1}^N r_i(\bm\nu).\\ \label{eq:Probplusminus_controlqubit_channel}
\end{align}
Following Eq.~\eqref{eq:cfi}, we calculate FI, $\mathcal F_\theta^q$ in NITDM  which have exactly same mathematical form as Eq.~\eqref{eq:FI_singleparam_contolqubitmeasure}, i.e., $\mathcal F_\theta^q\equiv\mathcal F_\theta$ with updated $r_i(\bm\nu)=\abs{\Tr \left( \sum_j K_j^{T} \rho_i K_j^{\dagger} \right)}$ and $f_i(\bm\nu)=\arg{\left[\Tr \left( \sum_j K_j^{T} \rho_i K_j^{\dagger} \right)\right]}$. With the tools in hand, we are now ready to analyze the effect of noisy encoding. 

\subsection{Comparison of NITDM with ``switched'' and regular strategy in single qubit $N=1$}
In the case of $N=1$, the value of FI  corresponding to NITDM is given by (see Appendix. \ref{app:nonoptimal_single})
\begin{eqnarray}
    \mathcal F_\theta^q = \frac{n_2^4 q^2\sin^2\theta}{1-\left\{\frac{1+q}{2}-n_2^2 q(1-\cos\theta)\right\}^2},
\label{eq:PhaseEstimation_ControlFI_Noise_TimeFlip}
\end{eqnarray}
which is independent of the input probe configuration. Being periodic in $\theta,$ we can safely choose $\theta\in[0,2\pi)$ for our analysis. We compare NITDM with two other metrological schemes -- $(1)$ ``switched'' strategy, where encoding is done using the quantum switch as described in Ref.~\cite{chapeau2021noisy}, $(2)$ regular strategy \cite{chapeau2021noisy} where conventional unitary encoding is used. The ``switched'' FI, $\mathcal{F}_\theta^{q,S}$ and QFI, $\mathcal{Q}_\theta^{q,R}$ in a regular  strategy are given by $\mathcal{F}_\theta^{q,S}={{\tilde q}^2\sin^4\theta}\left[1-\left\{\tilde q\cos\theta+\frac{1}{4}(1+q)^2\right\}^2\right]^{-1}$ with $\tilde q=q(1-q)$, and $\mathcal{Q}_\theta^{q,R} = q^2(\vec{n} \cross \vec{s})^2$ respectively. Let us now compare these three strategies in terms of their metrological power.

\subsubsection{Fixing optimally encoded direction} 

Notice that $\mathcal{F}_\theta^q$ is optimal for encoding axis $\hat n=\hat y$ (see Appendix.~\ref{app:nonoptimal_single}) where optimal QFI for regular probe $\mathcal{Q}_\theta^{q,R} = q^2$ is attained  by choosing $\hat s\perp \hat n$. Due to the $\theta$-dependency of the FI in the NITDM, there exist specific regions of $\theta$ where this protocol demonstrates superior performance compared to the regular one, as well as regions where it lags behind. Specifically, $\mathcal{F}_\theta^q|_{\hat n=\hat y}$ provides TF  advantage compared to $\mathcal{Q}_\theta^{q,R}|_{\hat s\perp \hat n} = q^2$ in the parameter region $\pi/3 < \theta < 
4 \tan^{-1}\left(\sqrt{q_1 - 
4 \sqrt{q_2}}\right)$ and $ 
4 \tan^{-1}\left(\sqrt{q_1 + 
4 \sqrt{q_2}}\right)<\theta<5\pi/3$ where $q_1=(5 + 7q)/(3 + q)$ and $q_2 = (1 + 4q + 3q^2)/(3 + q)^2$. In contrast, in the regime $\cos^{-1}\left[\frac{-1 + q}{2(1 + q)}\right] < \theta < 2\pi - \cos^{-1}\left[\frac{-1 + q}{2(1 + q)}\right]$; $0\leq \theta<\pi/3$ and $5\pi/3<\theta<2\pi/3$, we have $\mathcal{F}_\theta^q|_{\hat n=\hat y} <\mathcal{Q}_\theta^{q,R}|_{\hat s\perp \hat n}$, i.e., regular strategy outperforms NITDM. 

Note that these ranges of $\theta$ are reliant on the strength of the noise parameter $q$. Interestingly, there exist intervals of $\theta$ where NITDM surpasses regular strategy irrespective of the strength of the noise. Our analysis reveals that, $\mathcal{F}_\theta^q|_{\hat n=\hat y} >\mathcal{Q}_\theta^{q,R}|_{\hat s\perp \hat n}\,\forall q$ when $\{\pi/3<\theta<\pi/2;3\pi/2<\theta<5\pi/3\}$ while $\mathcal{F}_\theta^q|_{\hat n=\hat y} <\mathcal{Q}_\theta^{q,R}|_{\hat s\perp \hat n}\,\forall q$ for $\{2\pi/3<\theta<4\pi/3;0\leq \theta<\pi/3;5\pi/3<\theta<2\pi/3\}$. Moreover, setting $\hat n=\hat y$ as the encoding axis, we always obtain the benefit over the switched scheme, i.e., $\mathcal{F}_\theta^q|_{\hat n=\hat y}\geq \mathcal F_\theta^{q,S}$ (Appendix.~\ref{app:nonoptimal_single}) for any value of $\theta$ and arbitrary noise strength $q$.

\textit{Arbitrary probe state.} Since $\mathcal{Q}_\theta^{q,R}$ is dependent on probe states, choosing nonoptimal probe states can reduce the performance of regular strategy whereas NITDM bypasses the probe dependency and illustrates supremacy as shown in Fig.~\ref{fig:singlequbit}.  In particular, we illustrate the advantage of time-flip by choosing probe state $\rho$ with $s_x=0.8,s_y=0.6$, $\mathcal F_\theta^q$ can be shown to outperform both $\mathcal F_\theta^{q,S}$ and $\mathcal{Q}_\theta^{q,R}$ with $\theta=\pi/4$.
\begin{figure}[h]
    \centering
    \includegraphics[width=0.45\textwidth]{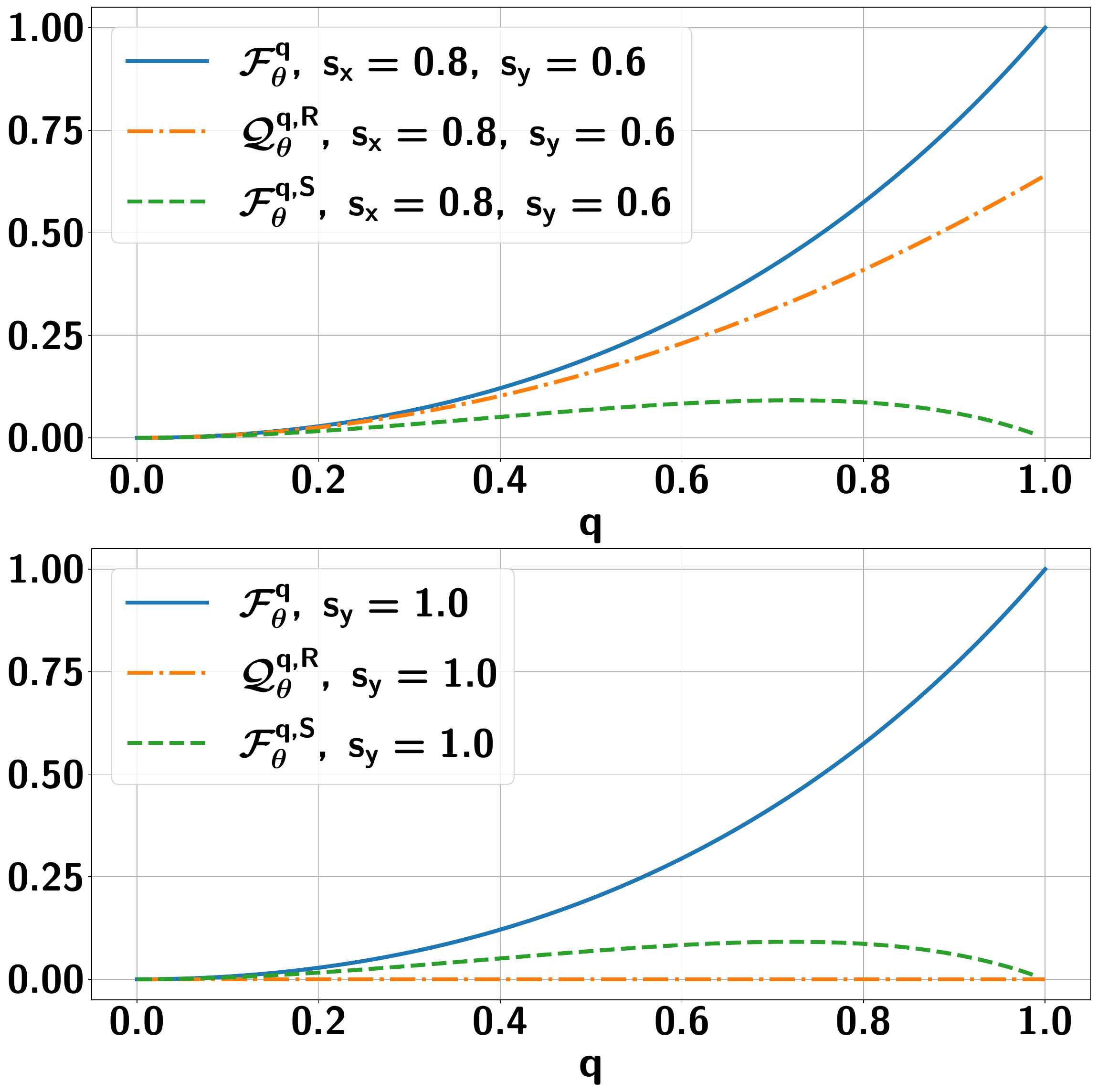}
    \caption{\textbf{Fisher information (vertical axis) in the case of time-flipped, \(F_\theta^q\), ``switched'', \(\mathcal F_\theta^{q,S}\) and regular protocols \(\mathcal{Q}_\theta^{q,R}\). with respect to noise parameter, $q$ (horizontal axis).}  Choosing $\hat n=\hat y$ with $\theta=\pi/4$, FI for NITDM, $\mathcal F_\theta^q$ (blue lines) always outperforms switched FI, $\mathcal F_\theta^{q,S}$ (green lines) which are independent of probe state. $\mathcal F_\theta^q$ is higher than $\mathcal{Q}_\theta^{q,R}$ (orange lines) with both $\hat s=[0.8,0.6,0]$ and $\hat s=\hat y$ probe state which are not optimal for regular strategy. All axes are dimensionless.}
    \label{fig:singlequbit}
\end{figure}
Interestingly, even when the probe is in a maximally mixed state or orthogonal to $\hat n$, $\mathcal F_\theta^q$ is always positive (except for $q=0$ and $\theta\neq 2m\pi\forall \pi\in\mathbb Z$) where $\mathcal{Q}_\theta^{q,R}=0$. Moreover, the dominance of NITDM over switched and regular strategies exists in the low noise region even with nonoptimal encoding as illustrated in Appendix.~\ref{app:nonoptimal_single}.

\subsubsection{$\theta$-averaged performance}
Due to the $\theta$-dependence of  $\mathcal{F}_{\theta}^q$ and $\mathcal F_{\theta}^{q,S},$ we perform averaging over $\theta$, i.e.,
\begin{eqnarray}
    \langle \mathcal{F}_\theta^{q}\rangle &=& \int_{0}^{2\pi}\frac{1}{2\pi}\mathcal  F_\theta^q d\theta \nonumber\\ &=& 1 -\frac{1}{4}\sqrt{(3+q)(3+q-4q n_2^2)} \nonumber\\ && \quad-\frac{1}{4}\sqrt{(1-q)(1-q+4qn_2^2)},
\end{eqnarray}
and  $\langle \mathcal{F}_\theta^{q,S}\rangle = 1-\frac{\sqrt 3}{8}(1-q)\sqrt{(1-q)(3+5q)}-\frac{1}{8}\sqrt{(5+6q-3q^2)(5-2q+5q^2)}$ \cite{chapeau2021noisy}. Given the noise strength $q$, although $\langle \mathcal{F}_\theta^{q}\rangle$ and $\langle \mathcal{F}_\theta^{q,S}\rangle$ are both less than the optimal regular QFI, $\mathcal{Q}_\theta^{q,R}|_{\hat n=\hat y}=q^2=\langle \mathcal{Q}_\theta^{q,R}|_{\hat n=\hat y}\rangle$, in the range $\pm n_2^{\min}\leq n_2\leq \pm1$, the ``flipped'' strategy is better than the ``switched'' one, i.e., $\langle \mathcal{F}_\theta^{q}\rangle>\langle \mathcal{F}_\theta^{q,S}\rangle$. Here, $n_2^{\min} = n_2^0[1 - q^g]^{h^{-1}}$ represents the $\chi^2$-fitted curve, where $g = 1.14147 \pm (0.48\%)$ and $h = 2.38455 \pm (0.25\%)$ are fitting parameters (see Appendix.~\ref{app:sing_avg}), and $n_2^0 = 0.945742$ denotes the value of $n_2^{\min}$ in the complete depolarizing scenario. This signifies another distinct advantage of the quantum time flip strategy over the switched one. In particular, in the presence of high noises, flip always outperforms switch when $\pm 0.945742\leq n_2\leq \pm1$, and when noise starts to decrease, this range gets increased with the noise parameter, making flip highly favorable in low noise regimes. Moreover, since the FI with switch method does not depend on $n_2$, for $\pm n_2^{\min}(q)\leq n_2\leq \pm1,$ the FI obtained with flip is strictly higher than the highest possible FI with the switch protocol at various levels of noise $q$. Because $n_2$ can be controlled,  we can obtain higher average FI values with time-flip for any value of $q$ compared to the switch-based scheme.

\subsection{Advantage in multiparty noisy case}

Going beyond the single qubit scenario, we examine how multiparty scenario gets affected by noise. While the optimal strategy for the multiqubit NITDM falls outside the scope of this study, we provide a brief overview of its key aspects. Note that while $\mathcal{F}_\theta^q$ is independent of the input probe for $N = 1$, this dependence emerges when $N \geq 2$. By assuming  $\theta=\pi/4, n_2=1$, we initialize the probe state to be $\ket{+i}^{\otimes N}$ and evaluate $\mathcal{F}_\theta^q$ (given by Eq.~\eqref{eq:FI_singleparam_contolqubitmeasure}) for different values of noise content $q$. Let us compare NITDM (by fixing $\theta=\pi/4$) with the regular strategy when all the $N$-qubit state is initialized along $\hat s=\hat y$, which is the optimal product probes in the noiseless case. In this case, $\mathcal Q^{q,R}_{\theta}=Nq^2$ while the behavior of $\mathcal{F}_{\theta = \frac{\pi}{4}}^q$ is depicted in Fig.~\ref{fig:multiqubit}. Due to $N^2$ term in Eq.~\eqref{eq:FI_singleparam_contolqubitmeasure}, $\mathcal F_\theta$ increases, while the exponential term is responsible for loss of Fisher information with $N$ as shown in Fig.~\ref{fig:multiqubit}. Specifically, depending on the system configuration and noise content, $\mathcal F_\theta^q$ reaches maximum at a certain value of $N$. However, FI per qubit,  $\mathcal F_\theta^q/N$, is maximized, where the first slope-change of the envelope of $\mathcal F_\theta^q$ occurs.
\begin{figure}[h]
    \centering
    \includegraphics[width=0.45\textwidth]{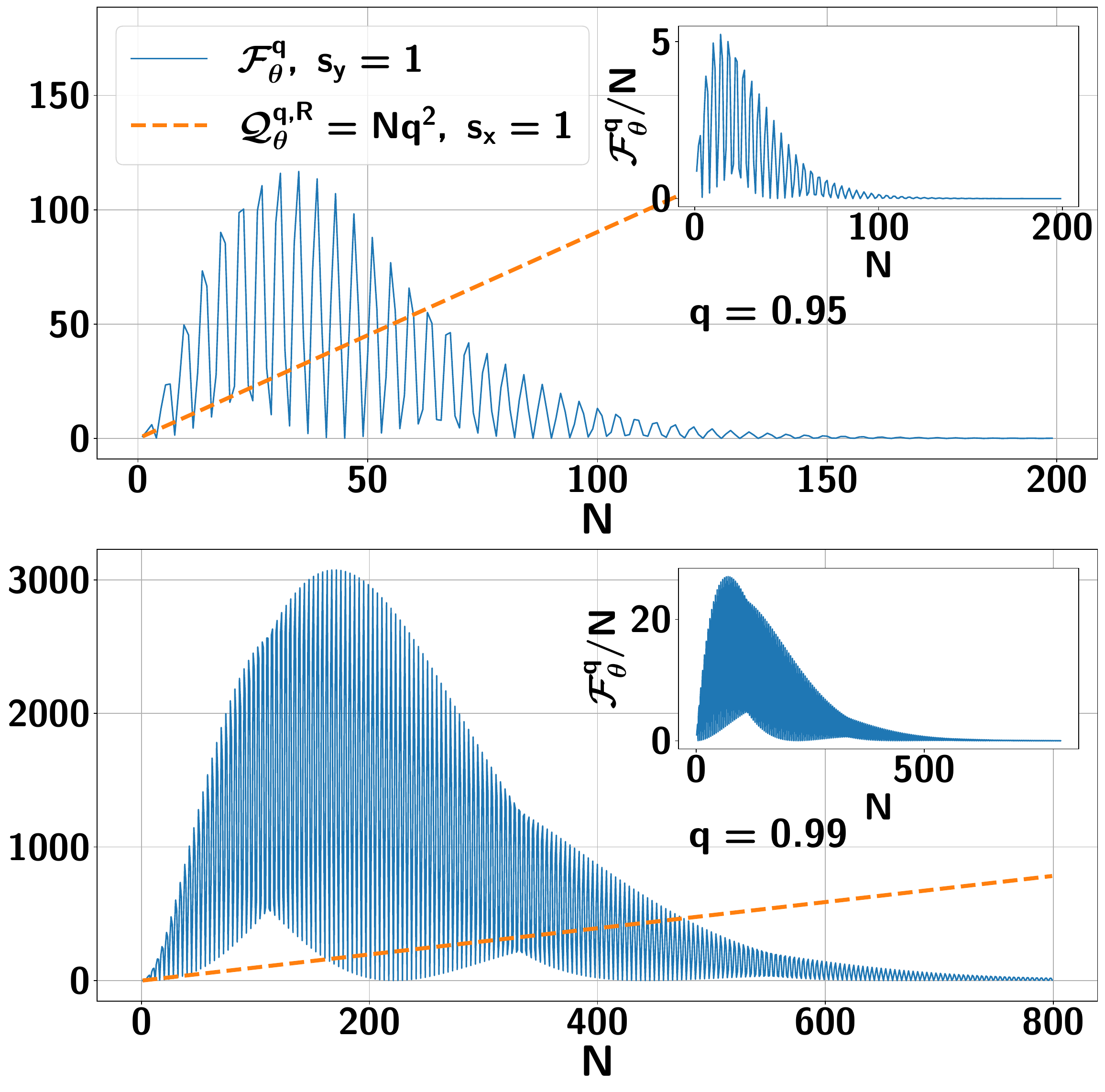}
    \caption{\textbf{Fisher information (ordinate) with respect to number or qubits $N$ (abscissa) for different values of noise parameter $q$. }  Fisher information for time-flip, $\mathcal{F}_\theta^q$ (blue lines) shows oscillatory behavior and after reaching the maximum, it decays. Here, we fix  the parameter values $\theta = \pi/4$ and $\hat n_2=1$. Contrastingly, QFI  for the regular case, $\mathcal{Q}_\theta^{(q,R)}$ (orange dotted lines) increases linearly. Inset: Fisher information per qubit $\mathcal{F}_\theta^q/N$ (ordinate, blue lines) against number or qubits $N$ (abscissa). All axes are dimensionless.}
    \label{fig:multiqubit}
\end{figure}
According to this, for a given $\theta,\vec s,$  and $\hat n$, we define the set $\{\mathcal{\tilde F}_\theta^q, \tilde N\}$  where $\mathcal{\tilde F}_\theta^q$ is the value of FI  of $\tilde N$-qubit NITDM protocol corresponding to the maximized value of $\mathcal F_\theta^q/N$  at $\tilde N$. Given a high value of $N\gg\tilde N,$  we perform  $\tilde N$-qubit NITDM and repeat the procedure for $\lfloor {N}/{\tilde N}\rfloor$ times to get an advantage over regular additive strategy with $N$-qubit product probe states. Note that if $N<\tilde N$, we have to perform $N$-qubit NITDM to get higher value of FI over the regular strategy.

To demonstrate the above claim, we find that maximum $\mathcal{F}_{\theta = \frac{\pi}{4}}^{q = 0.99} = 3074.72$ and $\mathcal{F}_{\theta = \frac{\pi}{4}}^{q = 0.95} = 116.77,$ which are achieved at \(  N = 170 \) and \( N = 34 \) in the NITDM scheme respectively. On the other hand, the maximum values of \( \mathcal{F}_\theta^q/N \) are $27.46$ and $5.51$,  achieved at $\{\mathcal{\tilde F}_\theta^q, \tilde N\}  = \{ 1894.5, 69\}$ and \( \{ 49.68, 9\} \) respectively. 
% Therefore, for a fixed high value of $N=1000$, for noise content $q=0.95$ performing NITDM with $N=34$ repeatedly $29$ times can achieve $29\mathcal{F}_\theta^{q=0.95}=29\times34=986$ which is significantly higher than $1000$-qubit regular QFI, $F^{q,R}_{\theta}|_{\hat n_2=1}=Nq^2=1000\times0.95^2=902.5$, i.e., beats SQL.
Let us elaborate on the benefit of NITDM by discussing two cases -- \textbf{Case 1.} For a fixed number of qubits, say $N = 1000$ and $q = 0.95$, employing the NITDM with $\tilde N = 9$ qubits for $111$ repetitions yields an effective Fisher information of $111\tilde{\mathcal{Q}}_\theta^{q=0.95}=111 \times 49.68 = 5514.48$, surpassing the QFI achievable with a standard $1000$-qubit system under the same noise conditions, $\mathcal Q^{(q=0.95,R)}_{\theta} = Nq^2 = 1000 \times 0.95^2 = 902.5$. \textbf{Case 2.} For $q=0.99$, we can perform NITDM  with $\tilde N=69$ repeatedly for $14$ times to get the best possible result $14\mathcal{F}_\theta^{q=0.99}=26523\gg\mathcal Q^{(q=0.99,R)}_{\theta} = 980.1$.

This shows a notable performance gain by NITDM compared to the regular strategy in a low-noise regime, although it is for a specific parameter $\theta$. However, our study provides evidence that this advantage persists when averaged over the parameter space of $\theta$ as previously defined, which we discuss in Appendix. \ref{app:avg_multiqubit}.

\section{Conclusion}
\label{sec:conclusion}

We incorporated the concept of an indefinite time directed encoding process in the framework of quantum metrology to enhance the precision limit of parameter estimation beyond the standard quantum limit (SQL). To accomplish the higher precision than the SQL, referred to as the Heisenberg limit (HL), the conventional parameter estimation strategy relies on an entangled input probe, along with unitary encoding and multi qubit  measurement  on the  encoded state for decoding. 
In contrast, we demonstrated that achieving the HL is possible by employing unentangled input probes with discrete energy levels provided indefinite time  directed encoding (ITDM) is performed.  Specifically, we proved analytically that the symmetric product probe state is sufficient to reach Heisenberg scaling among all possible multiparty input states although  
entangled input probe states can increase the precision of parameter estimation compared to  the separable probe ones.
Further, in certain cases involving product probes, the bidirectionally encoded state remains a product state, indicating that entanglement is not generated throughout the protocol even after the global time-flip operation and, therefore, is unnecessary. 
%We show that measuring only the control qubit of the final encoded state provides Heisenberg limited precision, and in certain scenarios, it attains the QFI. 
Fixing the measurement  on the control qubit, we illustrated that in the limit of small parameter values, both the input probe and quantum Fisher information (QFI) which provides a lower bound on the precision are independent of the parameters involved. We also examined the impact of noise on a single-qubit scenario and assessed its advantages in certain cases over the ``switched'' (based on indefinite causal order) and conventional metrology protocols. Additionally, we discussed strategies for achieving higher Fisher information in noisy multiqubit scenarios, in comparison to the regular scheme.

It is important to note that a quantum time flip operation cannot be realized within the standard quantum circuit model since no such circuit can transform an unknown unitary gate to its transpose \cite{chiribella2022quantum,Quintino2019}. However,  it was  demonstrated~\cite{guo2024experimental,stromberg2024experimental} that the implementation of time-flip operations on photonic platforms can be achieved by leveraging device-dependent symmetries inherent to their experimental apparatus. Specifically, the target qubit, subject to the time flip operation,  is encoded in the polarization degree of freedom of light, while its spatial modes serve the purpose of the control qubit. Forward and backward temporal directions are achieved by transmitting photons through appropriately configured waveplates, following the proper Stokes-parameter convention \cite{Frigo2022}. Consequently, given $N$ queries of the black boxes, $U_\theta$ (where $\theta$ is the parameter to be estimated), the ITDM scheme should be within experimental reach. Further, this can be used to validate the Heisenberg limit of root mean square error utilizing a bidirectional encoding process without entangled probes, which are typically fragile against noise.

Our results pave the way for the metrological utilization of quantum time flips in the discrete-variable regime. Going forward, it will be fascinating to explore the behavior of  QFI  in the presence of noise using other estimation procedures such as parallel, adaptive, and indefinite-causal-order approaches
%for scenarios involving a finite number of particles in the probe 
~\cite{liu2023optimal,mothe2024reassessing,Liu2024Dec} in conjunction with indefinite time direction. 
To determine whether time flip or more generic indefinite time directed \cite{chiribella2022quantum}  techniques outperform typical quantum comb or indefinite-causal-order systems, their hierarchies must be rigorously classified and evaluated possibly using semidefinite optimization techniques \cite{liu2023optimal,Kurdziaek2023,mothe2024reassessing,Liu2024Dec,Kurdziaek2025}.

\begin{comment}
\ga{What we have devised above is an exact universal strategy to gain $N^2$ advantage in quantum Fisher information for any possible channel and any possible set of parameters for arbitrary dimension of qudits unless and until they satisfy some very particular equations. Notably we also do not need to employ entangled states or measurements and even the generation of entanglement is not necessary at any point in circuit. The input states are always pure product states and the measurement can be done over a single control qubit making the advantage and heisenberg limit largely independent of the parameters which are encoded. The most this strategy needs are two qubit operations.
\end{comment}

\begin{comment}
From Theorem~\ref{th1} we can conclude that:
\begin{corollary}
    Maximal QFI in case of ITDM is achieved by symmetric states of the form $\ket{\Psi^S}=\ket{\psi_c}\ket{\psi}^{\otimes N}$ where we  have to optimize over only the doubly-parametrized set of states, $\ket{\psi}=\cos{\frac{\theta}{2}}\ket{0} + e^{i\phi}\sin{\frac{\theta}{2}}\ket{1}$ with $\theta\in[0,\pi]$ and $\phi\in[0,2\pi]$. 
\end{corollary}
\end{comment}

\acknowledgements
We acknowledge the use of \href{https://github.com/titaschanda/QIClib}{QIClib} -- a modern C++ library for general purpose quantum information processing and quantum computing and cluster computing facility at Harish-Chandra Research Institute. PH acknowledges ``INFOSYS
scholarship for senior students''.

\bibliographystyle{quantum}
\bibliography{reference}
%\addbibresource{reference.bib}

\onecolumn
\appendix
\section{Calculation of QFI for a arbitrary fully separable pure state}
\label{app:QFI_calc}

By using the form of pure product probe states, given in Eq.~\eqref{eq:PureStateOutput} of the main text, we obtain
\begin{eqnarray}
    \nonumber \dot{\ket{\Phi}}= \diff{\ket\Phi}{\theta} &=& \sqrt{p_c}\ket{0}\otimes\sum_{i=1}^N\bigg(\dot{\ket{\chi_i}}\bigotimes_{\substack{j=1\\j\neq i}}^N\ket{\chi_j}\bigg) + e^{i\theta_c}\sqrt{1-p_c}\ket{1}\otimes\sum_{i=1}^N\bigg(\dot{\ket{\omega_i}}\bigotimes_{\substack{j=1\\j\neq i}}^N\ket{\omega_j}\bigg).
\end{eqnarray}
Therefore, we can write the inner product of $\dot{\ket{\Phi}}$ with itself as
\begin{eqnarray}
    \langle \dot{\Phi} | \dot{\Phi} \rangle &=& p_c \left( \sum_{k=1}^N \langle \dot{\chi}_k | \dot{\chi}_k \rangle + \sum_{\substack{j,k=1\\ j \neq k}}^N \langle \dot{\chi}_k | \chi_k \rangle \langle \chi_j | \dot{\chi}_j \rangle \right) + (1 - p_c) \left( \sum_{k=1}^N \langle \dot{\omega}_k | \dot{\omega}_k \rangle + \sum_{\substack{j,k=1\\ j \neq k}}^N \langle \dot{\omega}_k | \omega_k \rangle \langle \omega_j | \dot{\omega}_j \rangle \right).\notag\\ \label{eq:QFI_firstterm_purestate1}
\end{eqnarray}
On the other hand, we have 
\begin{eqnarray}
    \langle \dot{\Phi} | \Phi \rangle = p_c\sum_{k=1}^N\langle \dot{\chi}_k | {\chi}_k \rangle + (1-p_c)\sum_{k=1}^N\langle \dot{\omega}_k | {\omega}_k \rangle.
    \label{eq:phidpdi}
\end{eqnarray}
Defining
\begin{align}
     & \quad \langle \dot{\chi}_k | \chi_k \rangle = \alpha_{k}, \quad \langle \dot{\chi}_k | \dot{\chi}_k \rangle = | \dot{\chi}_k |, \notag\\
& \quad \langle \dot{\omega}_k | \omega_k \rangle = \beta_{k}, \quad \langle \dot{\omega_k} | \dot{\omega_k} \rangle = |\dot{\omega_k} |,
\end{align}
we can rewrite Eqs.~\eqref{eq:QFI_firstterm_purestate1} and \eqref{eq:phidpdi} as
\begin{align}
\langle \dot{\Phi} | \dot{\Phi} \rangle & = p_c \left( \sum_{k=1}^N| \dot{\chi}_k | + \sum_{\substack{j,k=1\\j \neq k}}^N \alpha_{k}\alpha_{j}^{*} \right) \notag + (1-p_c) \left( \sum_{k=1}^N |\dot{\omega_k} | + \sum_{\substack{j,k=1\\j \neq k}}^N \beta_{k}\beta_j^{*} \right), \label{eq:QFI_firstterm_purestate2}
\end{align}
and
\begin{eqnarray}
    \langle \dot{\Phi} | \Phi \rangle =  p_c \sum_{k=1}^N \alpha_{k} + (1 - p_c) \sum_{k=1}^N \beta_{k},
    \label{eq:QFI_secondterm_purestate0}
\end{eqnarray}
respectively. Finally, we calculate the QFI, corresponding to an arbitrary product probe state as shown in  Eq.~\eqref{eq:symmproveeqnQ0} of main text.

\section{Proof of Theorem~\ref{th1}}
\label{ap:th1}
\begin{IEEEeqnarray}{rCl}
\nonumber&&\mathcal{Q}_{\text{avg}} - \mathcal{Q}^A 
    \\&=& 4 \Bigg[ p_c \left( \sum_{k =1}^N |\dot{\chi_k}| + (N-1) \sum_{k =1}^N \alpha_k \alpha_k^{*} \right)  + (1 - p_c) \left( \sum_{k =1}^N |\dot{\omega_k}| + (N-1) \sum_k \beta_k \beta_k^{*} \right) - N \sum_{k =1}^N \Big| p_c \alpha_k + (1 - p_c) \beta_k \Big|^2\Bigg] \notag\\
    &&- 4 \Bigg[ p_c \left( \sum_{k =1}^N |\dot{\chi_k}| + \sum_{k =1}^N \sum_{j \neq k}^N \alpha_k \alpha_j^{*} \right)  + (1 - p_c) \left( \sum_{k =1}^N |\dot{\omega_k}| + \sum_{k =1}^N \sum_{j \neq k}^N \beta_k \beta_j^{*} \right)  -  \bigg| p_c \sum_{k =1}^N\alpha_k + (1 - p_c) \sum_{k =1}^N\beta_k \bigg|^2\Bigg] \notag\\
    &=& 4 \Bigg[ p_c \left( (N-1) \sum_{k =1}^N \alpha_k \alpha_k^{*} - \sum_{k =1}^N \sum_{j \neq k}^N \alpha_k \alpha_j^{*} \right)  + (1 - p_c) \left( (N-1) \sum_{k =1}^N \beta_k \beta_k^{*} - \sum_{k =1}^N \sum_{j \neq k}^N \beta_k \beta_j^{*} \right) \notag \\
    && + \bigg| p_c \sum_{k =1}^N\alpha_k + (1 - p_c) \sum_{k =1}^N\beta_k \bigg|^2 - N \sum_{k =1}^N \Big| p_c \alpha_k + (1 - p_c) \beta_k \Big|^2 \Bigg] \notag \\
    &=& 2p_c \sum_{k =1}^N\sum_{j=1}^N \left( \alpha_k \alpha_k^{*} + \alpha_j \alpha_j^{*} - \alpha_k \alpha_j^{*} - \alpha_j \alpha_k^{*} \right)  + 2(1 - p_c) \sum_{k =1}^N\sum_{j=1}^N \left( \beta_k \beta_k^{*} + \beta_j \beta_j^{*} - \beta_k \beta_j^{*} - \beta_j \beta_k^{*} \right) \notag \\
    && + 4 \sum_{k =1}^N\sum_{j=1}^N \left( p_c^2 \alpha_k \alpha_j^{*} + (1 - p_c)^2 \beta_k \beta_j^{*}  + p_c (1 - p_c) \beta_k \alpha_j^{*}  + p_c (1 - p_c) \alpha_k \beta_j^{*}\right) \notag \\\nonumber
    && - 4N \sum_{k =1}^N \left( p_c^2 \alpha_k \alpha_k^{*} + (1 - p_c)^2 \beta_k \beta_k^{*}+ p_c (1 - p_c) \alpha_k \beta_k^{*} + p_c (1 - p_c) \beta_k \alpha_k^{*} \right)
\end{IEEEeqnarray}

\begin{eqnarray}
   \nonumber &=& 4N p_c (1 - p_c)  \sum_{k =1}^N \alpha_k \alpha_k^{*} + 4N p_c (1 - p_c)  \sum_{k =1}^N \beta_k \beta_k^{*}  - 4 p_c (1 - p_c) \sum_{k =1}^N\sum_{j=1}^N \alpha_k \alpha_j^{*} - 4 p_c (1 - p_c) \sum_{k =1}^N\sum_{j=1}^N \beta_k \beta_j^{*} \notag \\\nonumber
    &&+ 4 p_c (1 - p_c) \sum_{k =1}^N\sum_{j=1}^N \alpha_k \beta_j^{*} + 4 p_c (1 - p_c) \sum_{k =1}^N\sum_{j=1}^N \beta_k \alpha_j^{*}  - 4N p_c (1 - p_c) \sum_{k =1}^N \alpha_k \beta_k^{*} - 4N p_c (1 - p_c) \sum_{k =1}^N \beta_k \alpha_k^{*}\\
    &=& 4N p_c (1 - p_c) \left( \sum_{k =1}^N\alpha_k \alpha_k^{*} + \sum_{k =1}^N \beta_k \beta_k^{*} - \sum_{k =1}^N \alpha_k \beta_k^{*} - \sum_{k =1}^N \beta_k \alpha_k^{*} \right) \notag \\
    && + 4 p_c (1 - p_c) \left( \sum_{k =1}^N\sum_{j=1}^N \alpha_k \beta_j^{*} + \sum_{k =1}^N\sum_{j=1}^N \beta_k \beta_j^{*}  - \sum_{k =1}^N\sum_{j=1}^N \alpha_k \alpha_j^{*} - \sum_{k =1}^N\sum_{j=1}^N \beta_k \beta_j^{*} \right) \nonumber\\
    &=& 4N p_c (1 - p_c) \sum_{k =1}^N \left| \alpha_k -\beta_k \right|^2 - 4 p_c (1 - p_c) \left| \sum_{k =1}^N \alpha_k -\beta_k\right|^2 \geq 0
\end{eqnarray}
where the last inequality is due to Cauchy–Schwarz inequality.

\section{Proof of $A+B\leq 1$  in Theorem~\ref{th:opt_product}}
\label{app:a+b}
The sum of coefficients $A$ and $B$ is calculated as
\begin{align}
    \nonumber &A+B \\\nonumber &= 1-n_3^2(1-2p_s)^2+4\{n_1\cos{\theta_s} + n_2(2p_c-1)\sin{\theta_s}\}\bigg[n_3(1-2p_s)\sqrt{p_s(p_s-1)}\{n_1\cos{\theta_s+n_2(2p_c-1)\sin{\theta_s}}\}\bigg].\\
\end{align}
To maximize $A+B$ with respect to the parameter $p_c$ we calculate 
\begin{eqnarray}
    \diffp{{(A+B)}}{{p_c}}=8 n_2 \sin{\theta_s} \bigg[ 
    n_3 (1 - 2 p_s) \sqrt{p_s(1-p_s)} 
    - 2 p_s (1 - p_s)  \big\{ n_1 \cos{\theta_s} + n_2 (-1 + 2 p_c) \sin{\theta_s} \big\} 
\bigg].
\end{eqnarray}
Solving $\diffp{{(A+B)}}{{p_c}}=0$ we get 
\begin{eqnarray}
    p_c=\frac{1}{2}  - \frac{1}{4n_2}\left\{
        2 n_1 \cot \theta_s + \frac{n_3 (-1 + 2 p_s) \csc \theta_s}{\sqrt{p_s(1-p_s)}}
    \right\}=p_c^o, 
\end{eqnarray}
which is the only critical point. Now, the $2$nd order derivative is given by 
\begin{eqnarray}
    \diffp{{(A+B)}}{{{p_c}^2}} = -32 n_2^2 p_s (1-p_s) \sin^2{\theta_s}.
    \label{eq:dd_a+b}
\end{eqnarray}
From Eq.~\eqref{eq:dd_a+b} it is clear that for $-1\leq n_2\leq 1$, $0<p_s<1$ and $\theta_s\neq2m\pi$ with $m\in\mathbb Z$, i.e., for all valid region of $n_1,n_2,n_3,p_s,\theta_s,p_c$ where $A\neq 0$, we have $\diffp{{(A+B)}}{{{p_c}^2}}<0$. This implies that $A+B$ is strictly concave in $p_c$ and $A+B|_{p_c=p_c^o}=1$ is the local maximum. Therefore, the global maximum value of $A+B$ is given by $A+B|_{p_c=p_c^o}=1$.

\section{Axis Estimation with ITDM}
\label{app:axis}
The $2\times2$ unitary matrix is given by
\begin{eqnarray}
U&=&\exp\left\{{-i \frac{\theta}{2}\big(\sin\phi \cos\xi~\sigma_x+ \sin\phi \sin\xi~\sigma_y+\cos\phi~\sigma_z\big)}\right\},\nonumber\\
&=&\begin{bmatrix}
\cos\left(\frac{\theta}{2}\right) - i \cos\phi \sin\left(\frac{\theta}{2}\right) & 
-i \cos\xi \sin\phi \sin\left(\frac{\theta}{2}\right) - \sin\phi \sin\xi \sin\left(\frac{\theta}{2}\right) \\
-i \cos\xi \sin\phi \sin\left(\frac{\theta}{2}\right) + \sin\phi \sin\xi \sin\left(\frac{\theta}{2}\right) & 
\cos\left(\frac{\theta}{2}\right) + i \cos\phi \sin\left(\frac{\theta}{2}\right)
\end{bmatrix},
\end{eqnarray}
where $\theta$ is the phase and $\phi,\xi$ are the parameters of the axis of rotation of the unitary matrix. We aim to estimate $\phi$. According to corollary~\ref{cor1}, we initialize the input state as $\ket{\Psi^S}=\ket{\psi_c}\ket{\psi}^{\otimes N}$ where $\ket{\psi}=\sqrt{p_s}\ket{0} + e^{i\theta_s}\sqrt{1-p_s}\ket{1}$ with $\theta_s\in[0,2\pi]$ and $p_s\in[0,1]$ and $\ket{\psi_c} = \sqrt{p_c} \ket{0} + e^{i \theta_c} \sqrt{1 - p_c} \ket{1}$. In this scenario, we find that

\begin{eqnarray}
\alpha = \langle \dot{\chi}|\chi\rangle &=& i \Bigg[ (1/2 -  p_s) \sin\phi \sin\theta  \sqrt{p_s(1-p_s)} \cos\theta_s  \Big( \sin\xi - \cos\theta \sin\xi \
+ \cos\phi \cos\xi \sin\theta \Big) \nonumber \\
&& \qquad +  \sqrt{p_s(1-p_s)} \sin\theta_s \Big( -\cos\xi +\cos\theta \cos\xi 
+ \cos\phi \sin\xi \sin\theta \Big) \Bigg],\nonumber\\
\beta = \langle \dot{\omega}|\omega\rangle &=& i \Bigg[ (1/2 -  p_s) \sin\phi \sin\theta+  \sqrt{ps(1-ps)} \cos\theta_s  \Big( \cos\theta \sin\phi -\sin\xi + \cos\phi \cos\xi \sin\theta \Big) \nonumber \\
&&  \qquad + \sqrt{ps(1-ps)} \sin\theta_s \Big( -\cos\xi + \cos\theta \cos\xi - \cos\phi \sin\xi \sin\theta \Big) \Bigg],\nonumber\\
|\dot{\chi}| = \langle \dot{\chi}|\dot{\chi}\rangle &=& \sin^2{\frac{\theta}{2}}, \quad |\dot{\omega}| = \langle \dot{\omega}|\dot{\omega}\rangle= \sin^2{\frac{\theta}{2}},\nonumber
\end{eqnarray}
where the dot represents the derivative with respect to $\phi$. This gives a QFI of \(\mathcal{Q}_\phi =AN^2 +BN,\) where
\begin{eqnarray}
A =  4  p_c(1-p_c)|\alpha-\beta|^2, \quad
B = 4 \Big[ p_c \Big(| \dot{\chi} | - |\alpha|^2 \Big) + (1 - p_c) \Big( | \dot{\omega} | - |\beta|^2 \Big) \Big].
\end{eqnarray}
By optimizing $A$ with respect to the state and unitary parameters we get ${\displaystyle\max_{p_c,p_s,\theta_s,\xi,\theta}{A} = 4}$ at $p_c = 1/2,~ p_s = 1/2,~ \theta_s = \pi,~ \xi = \pi/2,~ \theta = \pi $ which in turn gives $B=0$. Thus, we get a quantum Fisher information of $\mathcal Q_\phi=4N^2$. It is easy to see that for any $AN^2 +BN$ and $CN^2+DN$ with $A \geq C$ implies $AN^2+BN \geq CN^2+DN$ for all positive integer value of $N$ iff $A+B \geq C+D$. Numerically optimizing $A+B$ indeed gives us a value of 4. Therefore the particular parameters employed provide us with the maximum value of QFI, $\mathcal{Q}_\phi=4N^2$ with optimal pure product states and optimal unitary parameters. In a similar manner, we can also calculate that optimal $\mathcal Q_\xi$ scales as $N^2$.

\section{Spectrum of $U_\theta^\dagger U_\theta^T$}
\label{app:udut}
$U_\theta^\dagger U_\theta^T$ is a unitary matrix, satisfying $U_\theta^\dagger U_\theta^T\ket{u(\hat n,\theta)^{\pm}} = \exp\left({\pm i \mathrm{f}(n_2, \theta)}\right)\ket{u(\hat n,\theta)^{\pm}}$. Here, $\mathrm f(n_2,\theta)$ is given by 
\begin{eqnarray}
    \mathrm f(n_2,\theta)=\tan^{-1}\left(\frac{\sqrt{4n_2^2\sin^2{\frac{\theta}{2}}\left(1-n_2^2 \sin^2{\frac{\theta}{2}}\right)}}{1-2n_2^2\sin^2{\frac{\theta}{2}}}\right).
\end{eqnarray}
The eigenvectors can be expressed as
\begin{eqnarray}
    \ket{u(\hat n,\theta)^{\pm}}=\frac{\mathrm g^\mp(\hat n, \theta)\ket{0}+\ket{1}}{\sqrt{1+|\mathrm g^\mp(\hat n, \theta)|^2}},
\end{eqnarray}
where 
\begin{eqnarray}
    \mathrm g^\mp(\hat n, \theta) = \frac{2n_1n_2 \sin^2{\frac{\theta}{2}}\mp\sqrt{4n_2^2\sin^2{\frac{\theta}{2}}\left(1-n_2^2 \sin^2{\frac{\theta}{2}}\right)}}{-2n_2n_3\sin^2{\frac{\theta}{2}}-in_2\sin\theta}.
\end{eqnarray}

\section{FI for small values of $\theta$}
\label{app:opt_probe_control}
For small values of $\theta=\tilde \theta,$ we consider up to $1$st order correction, i.e., by ignoring $\mathcal{O}(\theta^2)$. This approximates Eqs.~\eqref{eq:r_nu} and ~\eqref{eq:f_nu} in the main text as $r(\bm\nu)\approx 1 - n_2^2\tilde\theta^2(1-s_y^2)/2$ and $f(\bm\nu) \approx n_2s_y\tilde\theta - n_2\tilde\theta^2(n_3s_x - n_1s_z)/2$. Consequently, we have $\dot r(\bm\nu) = n_2^2 \tilde\theta(1-s_y^2)$ and $\dot f(\bm\nu) = n_2s_y - n_2\tilde\theta(n_3 s_x - n_1 s_z)$. Putting these values alongwith $\theta_c\neq 0$ in Eq.~\eqref{eq:FI_singleparam_contolqubitmeasure} in the main text  we get
\begin{eqnarray}
\mathcal F_{\tilde \theta} &\approx& N^2 \bigg(1 - \frac{n_2^2(1-s_y^2)\tilde \theta^2}{2}\bigg)^{2N-2} \Bigg[ -\tilde \theta n_2^2(1-s_y^2) \cos \Big( N\{ n_2s_y\tilde \theta + \mathcal{O}(\theta^2)\}+\theta_c\Big) -\bigg(1 - \frac{n_2^2(1-s_y^2)\tilde \theta^2}{2}\bigg) \nonumber\\
&&\nonumber \times\bigg(n_2s_y + \mathcal{O}(\theta)\bigg) \sin \Big( N \{n_2s_y\tilde \theta+\mathcal{O}(\theta^2)\} +\theta_c\Big)   \Bigg]^2 \Bigg/\Bigg[1 -  \bigg(1 - \frac{n_2^2(1-s_y^2)\tilde \theta^2}{2}\bigg)^{2N}\\\nonumber
&&\times\cos^2 \Big( N\{ n_2s_y\tilde \theta + \mathcal{O}(\theta^2)\}+\theta_c\Big)\Bigg], \nonumber\\
&\approx& N^2\Bigg[ -\tilde \theta n_2^2(1-s_y^2)\Big(\cos{\theta_c}-\{Nn_2s_y\tilde \theta+\mathcal{O}(\theta^2)\}\sin{\theta_c}\Big) -\bigg(1 - \frac{n_2^2(1-s_y^2)\tilde \theta^2}{2}\bigg)\bigg(n_2s_y +\mathcal{O}(\theta)\bigg) \nonumber\\ 
&&\times \bigg(\sin\theta_c+N\{n_2s_y\tilde \theta+\mathcal{O}(\theta^2)\}\cos{\theta_c}-\frac{\sin\theta_c}{2}N^2 \{n_2s_y\tilde \theta+\mathcal{O}(\theta^2)\}^2\bigg)\Bigg]^2 \Bigg/\Bigg[1 -  \bigg(1 - Nn_2^2(1-s_y^2)\tilde \theta^2\bigg) \nonumber\\
&&\times \bigg(\cos^2\theta_c - N\{n_2s_y\tilde \theta+\mathcal{O}(\theta^2)\}\sin2\theta_c-N^2\{n_2s_y\tilde \theta+\mathcal{O}(\theta^2)\}^2\cos2\theta_c\bigg)\Bigg], \nonumber\\
&\approx& \frac{N^2n_2^2s_y^2\sin^2\theta_c}{1-\cos^2\theta_c} = s_y^2n_2^2N^2.  
\label{eq:controlqubit_FI_theta0_thetacnonzero}
\end{eqnarray}

On the other hand, with $\theta_c=0$ we have
\begin{eqnarray}
    \mathcal F_{\tilde \theta} &\approx& N^2 \bigg(1 - \frac{n_2^2(1-s_y^2)\tilde \theta^2}{2}\bigg)^{2N-2} \Bigg[ -\tilde \theta n_2^2(1-s_y^2) \cos \Big( N n_2s_y\tilde \theta + \mathcal{O}(\theta^2)\Big) -\bigg(1 - \frac{n_2^2(1-s_y^2)\tilde \theta^2}{2}\bigg) \nonumber\\
&& \times\bigg(n_2s_y + \mathcal{O}(\theta)\bigg) \sin \Big( N n_2s_y\tilde \theta+\mathcal{O}(\theta^2) \Big)   \Bigg]^2 \Bigg/\Bigg[1 -  \bigg(1 - \frac{n_2^2(1-s_y^2)\tilde \theta^2}{2}\bigg)^{2N}\cos^2 \Big( N\{ n_2s_y\tilde \theta + \mathcal{O}(\theta^2)\}\Big)\Bigg], \nonumber\\
&\approx& \frac{N^2 \left( -\tilde \theta n_2^2(1-s_y^2) -Nn_2^2s_y^2\tilde \theta  - \mathcal{O}(\theta^2)\right)^2 }{1 -  (1 - Nn_2^2(1-s_y^2)\tilde \theta^2)(1-N^2n_2^2s_y^2\tilde \theta^2)}, \nonumber\\
&\approx& \frac{N^2n_2^4\tilde \theta^2 \left( (1-s_y^2) +Ns_y^2 \right)^2 }{Nn_2^2\tilde \theta^2((1-s_y^2)+Ns_y^2)} = Nn_2^2\left( (1-s_y^2) +Ns_y^2 \right) = s_y^2n_2^2N^2 +(1-s_y^2)n_2^2N.
\label{eq:controlqubit_FI_theta0_thetaczero}
\end{eqnarray}
Comparing Eqs.~\eqref{eq:controlqubit_FI_theta0_thetacnonzero} and ~\eqref{eq:controlqubit_FI_theta0_thetaczero}, it is evident that $\theta_c=0$ is the optimal scenario. However, $\theta_c\neq 0$ scenario also attains Heisenberg scaling.

\section{Proof of Theorem~\ref{th_ent}}
\label{app:th2}
A general $N$-qubit state can be written as 
\begin{eqnarray}
    \ket{\Psi} = \sum_{i=1}^{2^N} \Lambda_i\ket{\lambda_i},
\end{eqnarray}
where $(i-1)$ is the decimal equivalent number of the $i$th basis vector $\ket{\lambda_i}$, represented in binary number. For example, the decimal equivalent number of the basis state $\ket{\lambda_2}=\ket{\bm 0}^{\otimes N-1}\otimes \ket{\bm 1}$ is $1$, where $\ket{\bm{0}}$ and $\ket{\bm 1}$ is just the representation of qubit basis. Decomposing $\ket{\lambda_i}$ as $\bigotimes_{j=1}^N\ket{\psi^i_{j}}$ where $\ket{\psi_j^i}$ is the $i$th basis state for $j$th qubit with $\braket{\psi_j^i}{\psi_j^k}=\delta_{ik}$ and $\sum_{i=1}^2\ketbra{\psi_j^i}{\psi_j^i}=I_2$ we rewrite
\begin{eqnarray}
    \ket{\Psi} = \sum_{i=1}^{2^N} \Lambda_i \bigotimes_{j=1}^N\ket{\psi^i_{j}}.
\end{eqnarray}
The encoded state can be written as
\begin{eqnarray}
    \ket{\Phi} = \sum_{i=1}^{2^N} \Lambda_{i} \left( \sqrt{p_c} \ket{0} \bigotimes_{j=1}^{N} \ket{\chi^i_{j}} + e^{i \theta_c} \sqrt{1 - p_c} \ket{1} \bigotimes_{j=1}^{N} \ket{\omega^i_{j}} \right),
\label{eqn:entangled_output}
\end{eqnarray}
where $\ket{\chi_j^i} = U_\theta\ket{\psi^i_j}$ and $\ket{\omega_j^i} = U_\theta^T\ket{\psi^i_j}$. To calculate the QFI, we find the derivative of Eq.~\eqref{eqn:entangled_output} above with respect to parameter $\theta$ which is given by
\begin{eqnarray}
\dot{\ket \Phi} = \sum_{i=1}^{2^N} \Lambda_i \Bigg[ \sqrt{p_c} |0 \rangle \otimes \bigg(\sum_{m=1}^N | \dot\chi^i_m \rangle \bigotimes_{\substack{j=1 j \neq m}}^N | \chi^i_j \rangle  \bigg)\notag 
\nonumber + e^{i \theta_c} \sqrt{1 - p_c} |1 \rangle \otimes  \bigg(\sum_{m=1}^N | \dot\omega^i_m \rangle \bigotimes_{\substack{j=1\\j\neq m}}^N | \omega^i_j \rangle  \bigg)  \Bigg].
\end{eqnarray}
Therefore, we have 
\begin{eqnarray}
    \langle \dot{\Phi} | \dot{\Phi} \rangle &=& \sum_{i=1}^{2^N}\sum_{k=1}^{2^N} \Lambda_i\Lambda_k^* \Bigg[ p_c  \sum_{m=1}^N \Big\{ \langle \dot\chi^k_{m} | \dot\chi^i_{m} \rangle \prod_{j \neq m} \langle \chi^k_{j} | \chi^i_{j} \rangle \notag  + \sum_{l \neq m} \langle \dot\chi^k_{m} | \chi^i_{m} \rangle \langle \chi^k_{l} | \dot\chi^i_{l} \rangle \prod_{j \neq m, l} \langle \chi^k_{j} | \chi^i_{j} \rangle\Big\} \notag \\ \nonumber &&+ (1-p_c) \sum_{m=1}^N \Big\{ \langle \dot\omega^k_{m} | \dot\omega^i_{m} \rangle \prod_{j \neq m} \langle \omega^k_{j} | \omega^i_{j} \rangle \notag  + \sum_{l \neq m} \langle \dot\omega^k_{m} | \omega^i_{m} \rangle \langle \omega^k_{l} | \dot\omega^i_{l} \rangle \prod_{j \neq m, l} \langle \omega^k_{j} | \omega^i_{j} \rangle \Big\} \Bigg], \\ \nonumber &=&\sum_{i=1}^{2^N}\sum_{k=1}^{2^N}\Lambda_i\Lambda_k^* \Bigg[ p_c \sum_{m=1}^{N} \Big( \dot\chi_{m}^{ki} \prod_{j \neq m} a_{j}^{ki} 
+ \sum_{l \neq m} \alpha_{m}^{ki} \alpha_{l}^{ik *} \prod_{j \neq m, l} a_{j}^{ki} \Big)\\ \nonumber &&+ (1 - p_c)\sum_{m=1}^{N} \Big( \dot\omega_{m}^{ki} \prod_{j \neq m} b_{j}^{ki} 
+ \sum_{l \neq m} \beta_{m}^{ki} \beta_{l}^{ik *} \prod_{j \neq m, l} b_{j}^{ki} \Big) \Bigg].
\label{eq:QFI_firstterm_entangledprobe}
\end{eqnarray}
where 
\begin{eqnarray}
\langle \dot\chi^k_{m} | \dot\chi^i_{m} \rangle &=& \dot\chi_{m}^{ki}, \quad \langle \chi^k_{j} | \chi^i_{j} \rangle = a_{j}^{ki}, \nonumber\\
\langle \dot\omega^k_{m} | \dot\omega^i_{m} \rangle &=& \dot\omega_{m}^{ki}, \quad \langle \omega^k_{j} | \omega^i_{j} \rangle = b_{j}^{ki}, \nonumber\\
\langle \dot\chi^k_{m} | \chi^i_{m} \rangle &=& \alpha^{ki}_{m}, \quad \langle \dot\omega^k_{j} | \omega_{j}^i \rangle = \beta^{ki}_{j}. 
\end{eqnarray}
Now, $|\langle \dot{\Phi} | {\Phi} \rangle|^2$ being a positive number we have
\begin{eqnarray}
    \nonumber\mathcal{Q}_{\theta} &\leq& 4 \sum_{i=1}^{2^N}\sum_{k=1}^{2^N} \Lambda_i\Lambda_k^* \Bigg[ p_c \sum_{m=1}^{N} \Big( \dot\chi_{m}^{ki} \prod_{j \neq m} a_{j}^{ki} 
+ \sum_{l \neq m} \alpha_{m}^{ki} \alpha_{l}^{ik *} \prod_{j \neq m, l} a_{j}^{ki} \Big)\\ \nonumber &&+ (1 - p_c) \sum_{m=1}^{N} \Big( \dot\omega_{m}^{ki} \prod_{j \neq m} b_{j}^{ki} 
+ \sum_{l \neq m} \beta_{m}^{ki} \beta_{l}^{ik *} \prod_{j \neq m, l} b_{j}^{ki} \Big) \Bigg], \\ \nonumber &&= {\tilde{\mathcal{Q}}_\theta}.
\label{eq:generalQFI_entangledprobe_appx1}
\end{eqnarray}
To simplify $\mathcal{\tilde Q}_\theta,$ we observe that the terms $\prod_{j \neq m} a_{j}^{ki} $, and $\prod_{j \neq m} b_{j}^{ki} $ are the products of the inner products of $N-1$ of the qubits. Since $U_\theta$ and $U_\theta^T$ are both unitaries they preserve inner products, i.e., $a_j^{ki}=b_j^{ki}=\braket{\psi^k_j}{\psi^i_j}=\delta_{ki}$. Hence, given a value of $m$, $\prod_{j \neq m} a_{j}^{ki} = \prod_{j \neq m} b_{j}^{ki}=0$, if the qubits in $\ket{\lambda_i}$ and $\bra{\lambda_k}$ differ in more than one position. Similarly, if the qubits in $\ket{\lambda_i}$ and $\bra{\lambda_k}$ differ in more than two positions then the terms $\prod_{j \neq m, l} b_{j}^{ki}=\prod_{j \neq m, l} a_{j}^{ki}=0$ for a given pair of values $m$ and $l$. 

We define $\mathcal{\tilde{Q}}_{\theta}^{ X}$, which involves contribution from $N$-qubit basis states where qubits differ in  $X$ number of positions only. Therefore, we can write $\mathcal{\tilde{Q}}_{\theta}=\sum_{X=0}^N\mathcal{\tilde{Q}}_{\theta}^{ X}$. However, from the above discussions, we have $\mathcal{\tilde{Q}}_{\theta}^{ X}=0$ $\forall X\geq 3$ which implies $\mathcal{\tilde{Q}}_{\theta}=\sum_{X=0}^2\mathcal{\tilde{Q}}_{\theta}^{ X}$. So we have to calculate only $\mathcal{\tilde{Q}}_{\theta}^{0},\mathcal{\tilde{Q}}_{\theta}^{1}$ and $\mathcal{\tilde{Q}}_{\theta}^{2}$.

\begin{itemize}
    \item \textbf{Upper bound of  $\mathcal{\tilde{Q}}_{\theta}^{2}$:} A pair of $N$-qubit basis states, $\ket{\lambda_i}$ and $\ket{\lambda_j}$  which differ in exactly two qubits is given by the condition $j = i \oplus 2^k \oplus 2^l \text{ with } k \neq l$. Now, we can write
    \begin{IEEEeqnarray}{rCl}
        \nonumber\tilde{\mathcal{Q}}_\theta^{2} &=& 4\sum_{i = 1}^{2^N}\sum_{k=1}^{N}\sum_{\substack{l=1 \\ l \neq k}}^{N}\Lambda_i\Lambda_{i\oplus 2^k\oplus 2^l}^*\Bigg[p_c\bigg(\alpha^{i\oplus 2^k, i}_k\alpha^{i, i\oplus d^l *}_l + \alpha^{i\oplus d^l, i}_l\alpha^{i, i\oplus 2^k *}_k \bigg)\\\nonumber &&+(1-p_c)\bigg(\beta^{i\oplus 2^k, i}_k\beta^{i, i\oplus d^l *}_l + \beta^{i\oplus d^l, i}_l\beta^{i, i\oplus 2^k *}_k \bigg)\Bigg].
    \end{IEEEeqnarray}

    Since for any complex number, $z,$ we have $z+z^*\leq 2|z|$ and $|\sum_iz_i|\leq \sum_i|z_i|$, we can write
    \begin{eqnarray}
        \tilde{\mathcal{Q}}_\theta^{2} &\leq& 4\sum_{i = 1}^{2^N}\sum_{k=1}^{N}\sum_{\substack{l=1 \\ l \neq k}}^{N}|\Lambda_i\Lambda_{i\oplus 2^k\oplus 2^l}^*|p_c\bigg(\bigg|\alpha^{i\oplus 2^k, i}_k\alpha^{i, i\oplus 2^l *}_l\bigg| + \bigg|\alpha^{i\oplus 2^l, i}_l\alpha^{i, i\oplus 2^k *}_k\bigg| \bigg)\nonumber \\ &+& 4\sum_{i = 1}^{2^N}\sum_{k=1}^{N}\sum_{\substack{l=1 \\ l \neq k}}^{N}|\Lambda_i\Lambda_{i\oplus 2^k\oplus 2^l}^*|(1-p_c)\bigg(\bigg|\beta^{i\oplus 2^k, i}_k\beta^{i, i\oplus 2^l *}_l\bigg| + \bigg|\beta^{i\oplus 2^l, i}_l\beta^{i, i\oplus 2^k *}_k \bigg|\bigg).
    \end{eqnarray}
    Now, if $\max|\alpha^{i\oplus 2^k, i}_k|=\mathcal A$ and $\max|\beta^{i\oplus 2^k, i}_k|=\mathcal B$, where maximization is done over $\ket{\lambda_i},\bra{\lambda_{i\oplus2^k}}$ and parameters of the unitary by which encoding is performed, then 
    \begin{eqnarray}
        \tilde{\mathcal{Q}}_\theta^{2} &\leq& 8(p_c\mathcal{A}^2 + (1-p_c)\mathcal{B}^2) \sum_{k=1}^{N}\sum_{\substack{l=1 \\ l \neq k}}^{N}\Bigg|\sum_{i = 1}^{2^N}\Lambda_i\Lambda_{i\oplus 2^k\oplus 2^l}^*\Bigg|.
        \label{eq:ab}
    \end{eqnarray}
    Defining $\vec\Lambda=\{\Lambda_1, \Lambda_2,\ldots,\Lambda_N\}$ as the coefficient vector and $\{k,l\}$-permuted coefficient vector as $\vec\Lambda_{kl}=\{\Lambda_{1\oplus2^k\oplus2^l}, \Lambda_{2\oplus2^k\oplus2^l},\ldots,\Lambda_{N\oplus2^k\oplus2^l}\}$ we can write $\Bigg|\sum_{i = 1}^{2^N}\Lambda_i\Lambda_{i\oplus 2^k\oplus d^l}^*\Bigg|=| \vec{\Lambda_{k,l}}\cdot\vec{\Lambda}| \leq \sqrt{|\vec{\Lambda} \cdot \vec{\Lambda}||\vec{\Lambda_{k,l}} \cdot \vec{\Lambda_{k,l}}|} = 1$. This implies
    \begin{eqnarray}
    \tilde{\mathcal{Q}}_\theta^{2} &\leq& 8(p_c\mathcal{A}^2 + (1-p_c)\mathcal{B}^2) \sum_{k=1}^{N}\sum_{\substack{l=1 \\ l \neq k}}^{N}1 = 8(p_c\mathcal{A}^2 + (1-p_c)\mathcal{B}^2)N(N-1).
    \label{eq:ab2}
    \end{eqnarray}

    \item \textbf{Upper bound of  $\mathcal{\tilde{Q}}_{\theta}^{1}$:} A pair of $N$-qubit basis states, $\ket{\lambda_i}$ and $\ket{\lambda_j}$  which differ in exactly one qubit is characterized by the condition $j = i \oplus 2^k$ which simplifies ${\mathcal{\tilde Q}_\theta}^{1}$ as
    \begin{eqnarray}
        {\mathcal{\tilde Q}_\theta}^{1} &=&4\sum_{i = 1}^{2^N}\sum_{k = 1}^{N}\Lambda_{i}\Lambda_{i \oplus 2^k}^*\Bigg[ p_c\bigg\{\dot\chi^{i \oplus 2^k, i}_k + \sum_{l \neq k} \alpha^{i\oplus 2^k, i}_k\alpha^{i, i\oplus 2^k *}_l +\sum_{m \neq k} \alpha^{i\oplus 2^k, i}_m\alpha^{i, i\oplus 2^k *}_k\bigg\} \nonumber\\
        && + (1-p_c)\bigg\{\dot\omega^{i \oplus 2^k, i}_k + \sum_{l \neq k} \beta^{i\oplus 2^k, i}_k\beta^{i, i\oplus 2^k *}_l+\sum_{m \neq k} \beta^{i\oplus 2^k, i}_m\beta^{i, i\oplus 2^k *}_k\bigg\}\Bigg].\\
        &\leq&4\Bigg[\sum_{i = 1}^{2^N}\sum_{k = 1}^{N}|\Lambda_{i}\Lambda_{i \oplus 2^k}^*| p_c\Bigg|\dot\chi^{i \oplus 2^k, i}_k\Bigg| + \sum_{i = 1}^{2^N}\sum_{k = 1}^{N}|\Lambda_{i}\Lambda_{i \oplus 2^k}^*| p_c\sum_{l \neq k} \Bigg|\alpha^{i\oplus 2^k, i}_k\alpha^{i, i\oplus 2^k *}_l\Bigg| \nonumber\\
        &&  + \sum_{i = 1}^{2^N}\sum_{k = 1}^{N}|\Lambda_{i}\Lambda_{i \oplus 2^k}^*| p_c\sum_{m \neq k} \Bigg|\alpha^{i\oplus 2^k, i}_m\alpha^{i, i\oplus 2^k *}_k\Bigg| + \sum_{i = 1}^{2^N}\sum_{k = 1}^{N}|\Lambda_{i}\Lambda_{i \oplus 2^k}^*| (1-p_c)\bigg|\dot\omega^{i \oplus 2^k, i}_k\Bigg| \nonumber\\\nonumber
        && + \sum_{i = 1}^{2^N}\sum_{k = 1}^{N}|\Lambda_{i}\Lambda_{i \oplus 2^k}^*| (1-p_c)\sum_{l \neq k} \Bigg|\beta^{i\oplus 2^k, i}_k\beta^{i, i\oplus 2^k *}_l\Bigg|  + \sum_{i = 1}^{2^N}\sum_{k = 1}^{N}|\Lambda_{i}\Lambda_{i \oplus 2^k}^*| (1-p_c)\sum_{m \neq k} \Bigg|\beta^{i\oplus 2^k, i}_m\beta^{i, i\oplus 2^k *}_k\Bigg|\Bigg].\\
    \end{eqnarray}

    Similar to Eqs.~\eqref{eq:ab} and \eqref{eq:ab2}, we can write 
    \begin{eqnarray}
    \tilde{\mathcal{Q}_\theta}^{(1)} &\leq& 4N(p_c\mathcal{X}+(1-p_c)\mathcal{W}) + 8(p_c\mathcal{A}^2 +(1-p_c)\mathcal{B}^2)N(N-1),
    \end{eqnarray}
    where $\mathcal X$ and $\mathcal{W}$ are upper bounds of $\left|\dot\chi^{i \oplus 2^k, i}_k\right|$ and $\left|\dot\omega^{i \oplus 2^k, i}_k\right|$ respectively. 

    \item \textbf{Upper bound of  ${\mathcal{\tilde Q}_\theta}^{0}$: } Since the basis states which differ in no position is that basis state itself,  the value of ${\mathcal{\tilde Q}_\theta}^{0}$ is upper bounded by 
    \begin{eqnarray}
        \tilde{\mathcal{Q}_\theta}^{0} &\leq& \bigg(A'N^2 +B'N\bigg) \sum_{i=1}^{2^N}\Lambda_i\Lambda_i^* = \bigg(A'N^2 +B'N\bigg), 
    \end{eqnarray}
    where $A'N^2+B'N$ is the upper bound of $\langle{\dot \Phi}|{\dot \Phi}\rangle$ for a 
    $N$-qubit product-probe state as evident from Sec.~\ref{sec:product_probe_itdm} in the main text. 
\end{itemize}
Finally, we have 
\begin{eqnarray}
    {\mathcal{\tilde Q}_\theta}={\mathcal{\tilde Q}_\theta}^{0}+{\mathcal{\tilde Q}_\theta}^{1}+{\mathcal{\tilde Q}_\theta}^{2}\leq \gamma N^2+\zeta N
\end{eqnarray}
where both $\gamma$ and $\zeta$ are some function of $\mathcal X, \mathcal W, \mathcal A, \mathcal B, A',B'$ only which are independent of $N$. Hence, we prove that no entangled state can beat Heisenberg scaling in ITDM. Note that, from the above calculation, it trivially follows that this scaling is also true for general $N$-qudit states.

\section{Demonstration of advantage in $N=1$ for nonoptimal encoding}
\label{app:nonoptimal_single}
In the case of NITDM by choosing $\rho_i=\rho\forall i$ as probe state and measuring in the $\{\ket +,\ket -\}$ basis on the control qubit, the FI is given by
\begin{eqnarray}
    \mathcal{F}_\theta=\frac{N^2  r^{2N-2}(\bm\nu) \left\{ \dot{r}(\bm\nu) \cos \left( N f(\bm\nu) + \theta_c \right) - r(\bm\nu) \dot{f}(\bm\nu) \sin \left( N f(\bm\nu) + \theta_c \right) \right\}^2 }{1 -  r^{2N}(\bm\nu) \cos^2 \left( N f(\bm\nu) + \theta_c \right)}.
    \label{eq:nitdm}
\end{eqnarray}
where $r(\bm\nu)=\abs{\Tr \left( \sum_j K_j^{T} \rho K_j^{\dagger} \right)}$ and $f(\bm\nu)=\arg{\left[\Tr \left( \sum_j K_j^{T} \rho K_j^{\dagger} \right)\right]}$. In case of $N=1$, Eq.~\eqref{eq:nitdm} reduces to
\begin{eqnarray}
    \mathcal{F}_\theta=\frac{\left\{ \dot{r}(\bm\nu) \cos \left(  f(\bm\nu) + \theta_c \right) - r(\bm\nu) \dot{f}(\bm\nu) \sin \left(  f(\bm\nu) + \theta_c \right) \right\}^2 }{1 -  r(\bm\nu) \cos^2 \left(  f(\bm\nu) + \theta_c \right)} = \frac{\Big\{\frac{d}{d\theta}\Big( r(\bm\nu) \cos \left(  f(\bm\nu) + \theta_c \right) \Big)\Big\}^2}{1-\Big(r(\bm\nu) \cos \left(  f(\bm\nu) + \theta_c \right)\Big)^2}
\label{eq:PhaseEstimation_noise_controlmeasure_plusstate_singleprobe}
\end{eqnarray}
From 
\begin{eqnarray}
\nonumber r(\bm\nu)=\abs{\Tr \left( \sum_j K_j^{T} \rho K_j^{\dagger} \right)} = \Bigg|\frac{1+q}{2}-n_2^2q(1-\cos\theta) +iq n_2\Big((1-\cos\theta)(n_1 r_3 - r_1 n_3)+ r_2\sin\theta \Big)\Bigg|,\\
\label{eq:Polarform_QuantumTimeFlip_phaseestimation_evaluation}
\end{eqnarray}
we have
\begin{eqnarray}
r(\bm\nu) \cos \left(  f(\bm\nu) + \theta_c \right) = \Re\left\{\Tr \left( \sum_j K_j^{T} \rho K_j^{\dagger} \right)\right\}  = \frac{1+q}{2}-n_2^2q(1-\cos\theta).
\label{eq:PolarformRealPart_QuantumTimeFlip_phaseestimation}
\end{eqnarray}
Putting Eq.~\eqref{eq:PolarformRealPart_QuantumTimeFlip_phaseestimation} in Eq.~\eqref{eq:PhaseEstimation_noise_controlmeasure_plusstate_singleprobe} we get 
\begin{eqnarray}
    \mathcal F_\theta^q = \frac{n_2^4 q^2\sin^2\theta}{1-\left\{\frac{1+q}{2}-n_2^2 q(1-\cos\theta)\right\}^2}.
\label{eq:PhaseEstimation_ControlFI_Noise_TimeFlip}
\end{eqnarray}

\textbf{Proof of optimality of $\hat n_2=\hat y$ and comparison with ``switched'' strategy.} Note that $\mathcal{F}_\theta^q$ can be written as
\begin{eqnarray}
    \nonumber\mathcal F_\theta^q = \frac{n_2^4 q^2\sin^2\theta}{1-\left[\left(\frac{1+q}{2}\right)^2+n_2^2 \{n_2^2q^2(1-\cos\theta)^2-q(1+q)(1-\cos\theta)\}\right]}.
    \label{eq:n2_N1}
\end{eqnarray}
From Eq.~\eqref{eq:n2_N1} it is clear that the numerator is maximized at $n_2=\pm 1$ where the denominator is minimized. Hence, $n_2=\pm 1$ is the optimal encoding axis in this scenario.
\begin{figure}[H]
    \centering
    \includegraphics[width=0.5\textwidth]{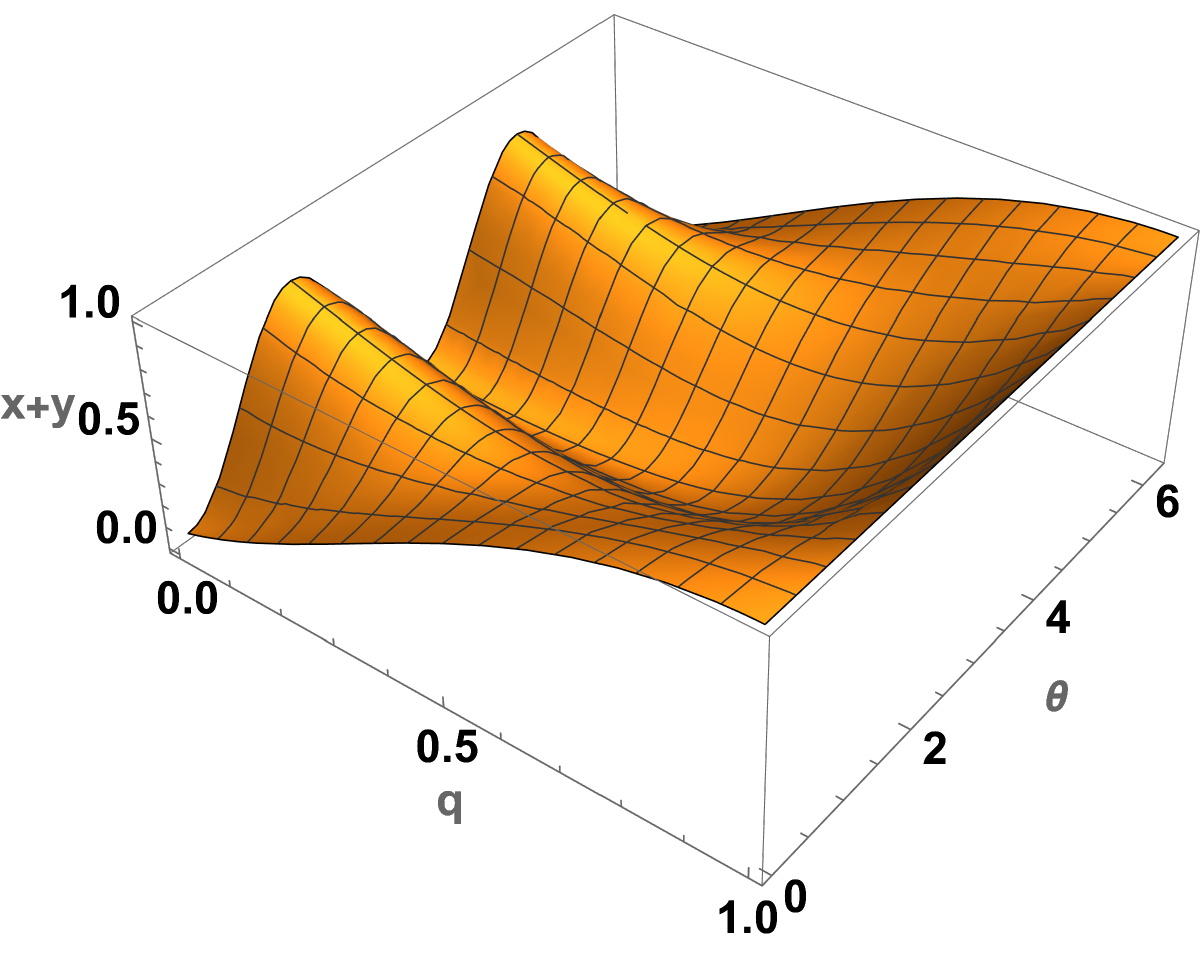}
    \caption{\textbf{Value of $x+y$ (ordinate) against noise strength $q$ (abscissa), and paramater $\theta$ (abscissa)} The figure demonstrates the value of $x+y$ never goes above 1. All axes are dimensionless.}
    \label{fig:fgs}
\end{figure} 
To compare with ``switched'' strategy let us evaluate the difference between $\mathcal F_\theta^q|_{\hat n=\hat y}$ and the ``switched'' strategy, defined as $\delta = \mathcal F_\theta^q|_{\hat n=\hat y}-\mathcal F_\theta^{q,S} $. 
Now, $\delta$ can be explicitly calculated as $\delta= 1-x-y$ where $x=\frac{1}{16} \left\{ (1 + q)^2 - 4(-1 + q)q \cos\theta \right\}^2$ and $y=(1 - q)^2 \sin^2\theta \left\{1 - \frac{1}{4} \left( -1 + q - 2q \cos\theta \right)^2 
\right\} $. Our numerical investigation suggests (see Fig.~\ref{fig:fgs}) that $x+y\leq 1$ which proves $\delta\geq 0\,\forall q,\theta$.

\textbf{Comparing NITDM with ``switched'' and regular strategy with non-optimal encoding. }Here, we demonstrate that, beyond the optimal encoding procedure, NITDM can also be advantageous in certain scenarios even under non-optimal encoding (see Fig.~\ref{fig:singlequbitnitdm_adv}), outperforming both the ``switched'' and regular strategies. Specifically, when selecting \(\hat{n} = [0.8, 0.6, 0]\) instead of \(n_2 = 1\), NITDM proves to be more effective than the ``switched'' strategy in the region \(q \gtrsim 0.73\), corresponding to the low-noise regime. 
\begin{figure}[H]
    \centering
    \includegraphics[width=0.5\textwidth]{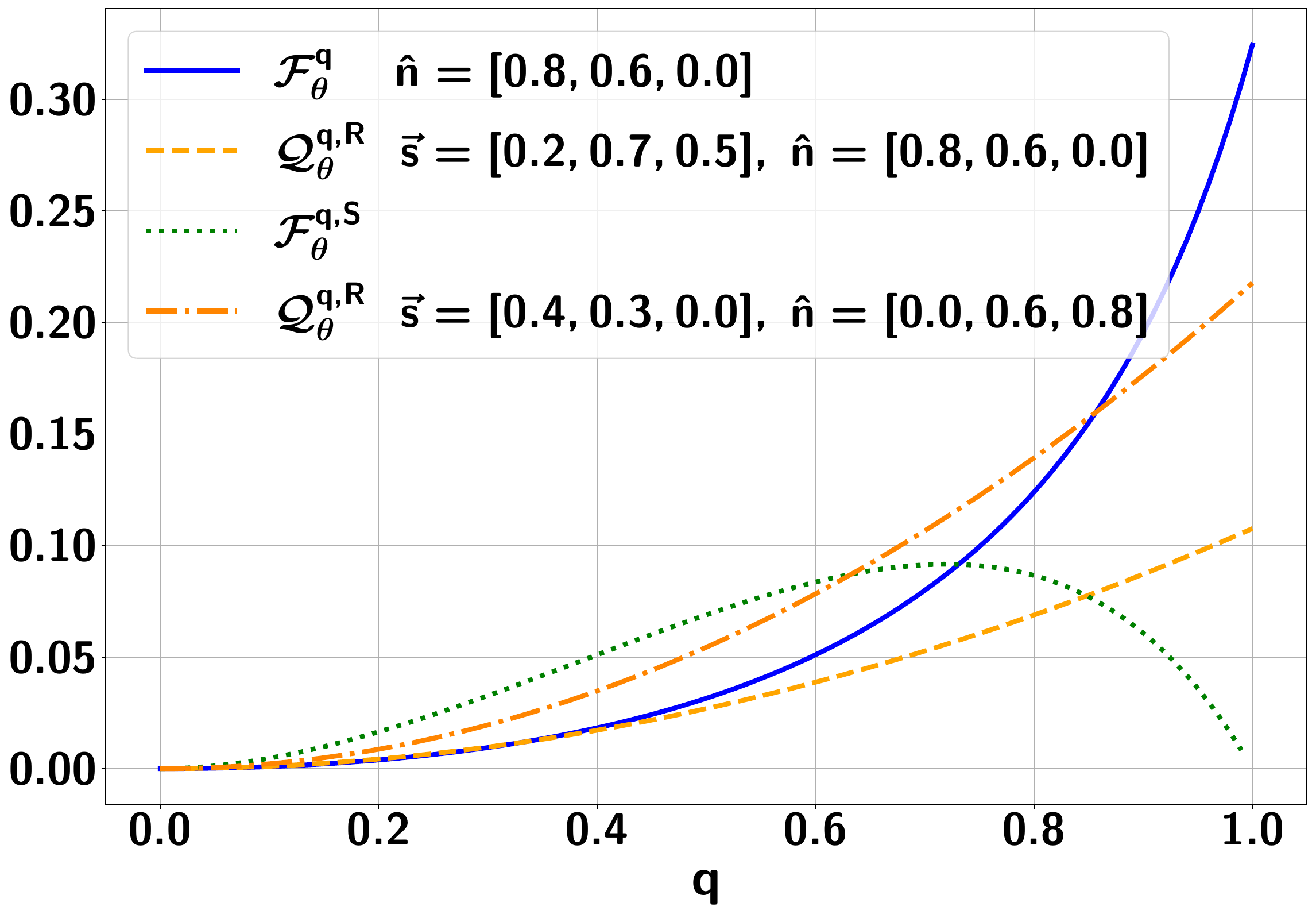}
    \caption{\textbf{Fisher information, $\mathcal F_\theta^q,\mathcal F_\theta^{q,S}$ and QFI, $\mathcal Q_\theta^{q,R}$ (ordinate) with respect to noise parameter $q$ (abscissa). Comparison between Flip, Switch, and regular strategies for $n_2 \neq 1$.} We choose $\theta = \pi/4$. The figure demonstrates the advantage of time flip over ``switched'' and regular strategies in non-ideal encoding scenarios. All axes are dimensionless.}
    \label{fig:singlequbitnitdm_adv}
\end{figure}
Conversely, maintaining \(n_2\) constant while varying \(n_1\), \(n_3\), and the probe state can reduce the advantage of the time-flip approach compared to regular strategy, as illustrated in Fig.~\ref{fig:singlequbitnitdm_adv}. While for $\hat n=[0.8,0.6,0]$ and $\hat s=[0.2,0.7,0.5]$ NITDM always excels regular strategy, in case of $\hat n=[0,0.6,0.8]$ and $\hat s=[0.4,0.3,0.0]$ only for $q\gtrsim 0.85$ $\mathcal F_\theta^q$ is higher. Nevertheless, in many cases within the low-noise regime, NITDM exhibits superior performance.

\section{$\theta$-averaged performance of single qubit-NITDM}
\label{app:sing_avg}
Since both $\mathcal F_\theta^q$ and $\mathcal F_\theta^{q,S}$ are $\theta$-dependent, to compare the performance of protocols with the aid of switch and flip, we evaluate the difference of the average FI obtained through these strategies, i.e., $\langle \mathcal{F}_\theta^{q}\rangle - \langle \mathcal{F}_\theta^{q, S}\rangle=\frac{\sqrt 3}{8}(1-q)\sqrt{(1-q)(3+5q)}+\frac{1}{8}\sqrt{(5+6q-3q^2)(5-2q+5q^2)} -\frac{1}{4}\sqrt{(3+q)(3+q-4q n_2^2)} -\frac{1}{4}\sqrt{(1-q)(1-q+4qn_2^2)}$. Analyzing this particular equation is nontrivial because of the associated square roots. Hence, we opt for a numerical approach instead. One thing we know is that for $n_2 = \pm1$, NITDM always outperforms the ``switched'' one at all $q$ and $\theta$ values. Therefore, it always beats the switched method on average for all $q$ values. For a given value of $q$  we analyze the range of $n_2$ values $\pm n_2^{\min}\leq n_2\leq \pm1$ where flip strategy is better, i.e., $\langle \mathcal{F}_\theta^{q}\rangle>\langle \mathcal{F}_\theta^{q,S}\rangle$ and fit $n_2^{\min}$ with, $n_2^{\min} = n_2^0[1 - q^g]^{h^{-1}}$. Fig. \ref{fig:n2minfit}  represents the $\chi^2$-fitted \cite{Pearson1900,Cochran1952} curve of $n_2^{\min}$ with $g = 1.14147 \pm (0.48\%)$ and $h = 2.38455 \pm (0.25\%)$ as fitting parameters and the constant $n_2^0 = 0.945742$. Although the magnitude of the advantage decreases significantly as \( n_2 \) and \( q \) deviate from unity, in low-noise regimes, metrology employing time-flip exhibits superior performance compared to that using switch over a broad range of \( n_2 \). 
\begin{figure}[H]
    \centering
    \includegraphics[width=0.45\textwidth]{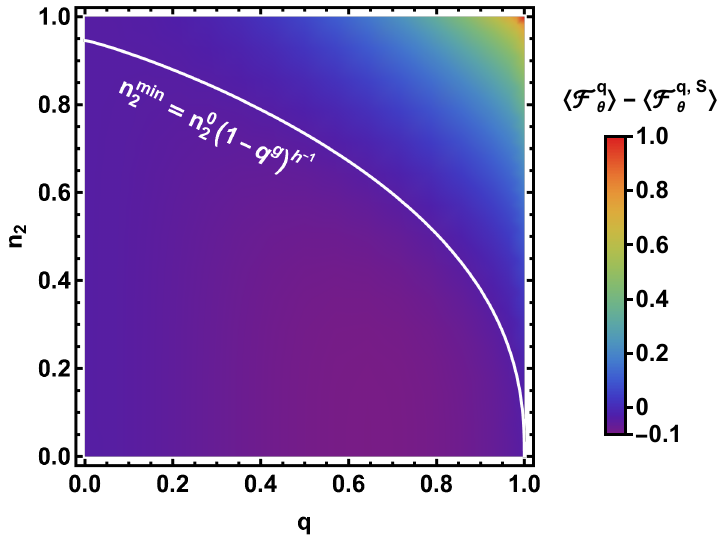}
    \caption{\textbf{Comparison between Flip and switch strategy on average. Variation of $\langle \mathcal{F}_\theta^{q}\rangle - \langle \mathcal{F}_\theta^{q, S}\rangle$ (heatmap) against $n_2$ (vertical axis) and  noise strength $q$ (horizontal axis). } The white contour is the $\chi^2$-fitted line of $n_2^{\min}$ where the difference $\langle \mathcal{F}_\theta^{q}\rangle - \langle \mathcal{F}_\theta^{q, S}\rangle=0$. Fitted curve is  $n_2^{\min} = n_2^0[1 - q^g]^{h^{-1}}$ with $g = 1.14147 \pm (0.48\%)$ and $h = 2.38455 \pm (0.25\%)$ as fitting parameters and a constant $n_2^0=0.945742$. In the regions above the curve  flip outperforms switch. All axes are dimensionless.}
    \label{fig:n2minfit}
\end{figure}

\section{$\theta$-averaged performance of multiqubit-NITDM}
\label{app:avg_multiqubit}
Similar to the case where $\theta$ is fixed, in  $\theta$-averaged scenario we define $\{\langle \mathcal{\tilde F}_\theta^q\rangle, \tilde N\}$  corresponding to which $\langle\mathcal F_\theta^q/N\rangle$ is maximized at $N=\tilde N$. Given $q=0.99$ and $0.95$, in Fig.~\ref{fig:multiqubitavg} we demonstrate the behavior of $\theta$-averaged FI with respect to the number of qubits and corresponding quantities are listed in Table~\ref{tab:multiqubitavg}.
\begin{figure}[H]
    \centering
    \includegraphics[width=0.45\textwidth]{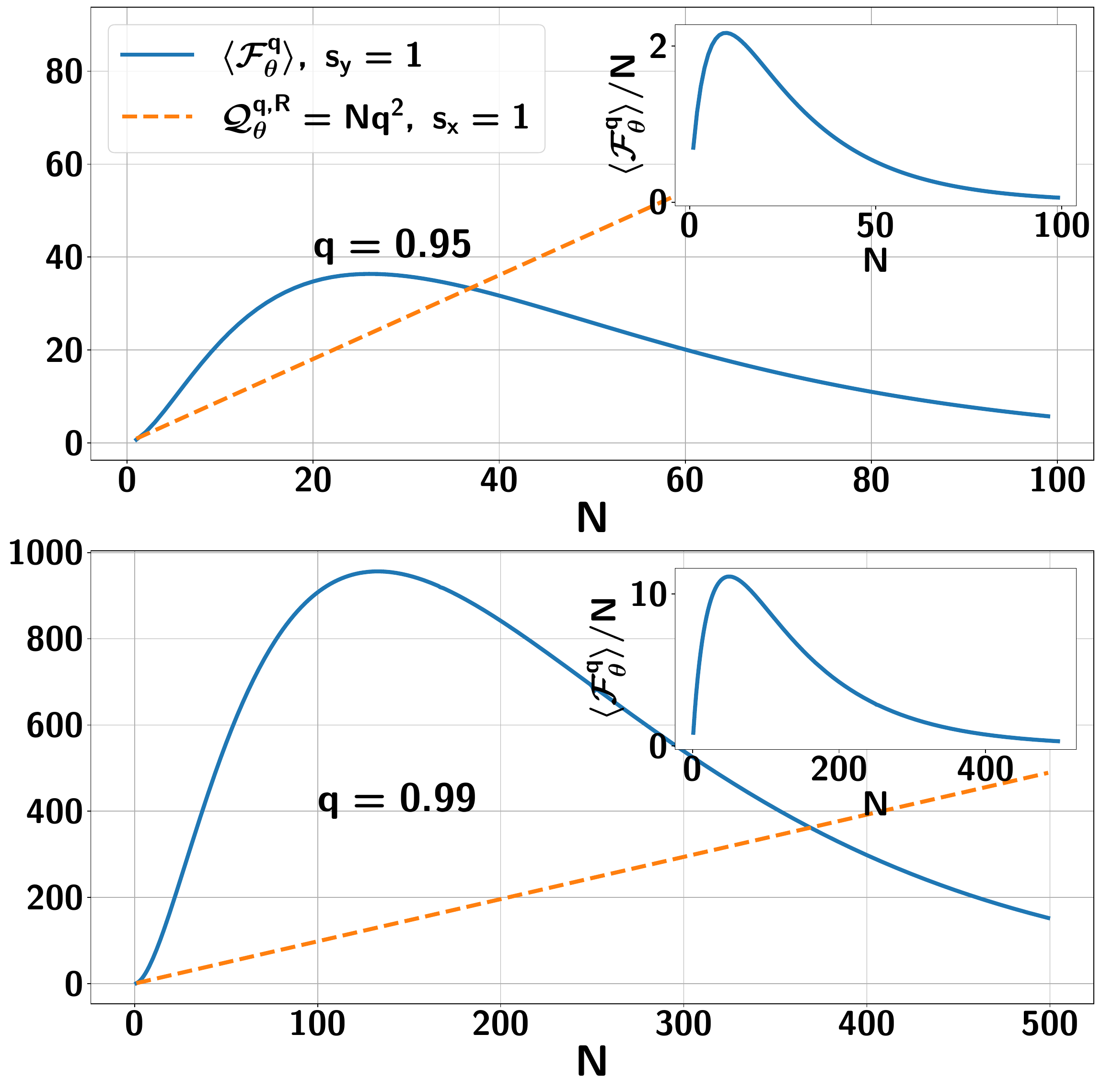}
    \caption{\textbf{ Advantage of average performance of multi qubit NITDM and regular strategy.} Average FI, $\langle \mathcal F_\theta^q\rangle$ (blue lines) with respect to number or qubits $N$ (abscissa) for different noise strength $q$. QFI in regular strategy, $\mathcal Q_\theta^{q,R}$ (orange dashed lines) against $N$. Inset: Average Fisher information per qubit, $\langle \mathcal F_\theta^q\rangle/N$ (ordinate) against number or qubits $N$ (abscissa). Here, we fix the parameter value $\theta = \pi/4$ and the encoding axis as $\hat n_2=1$.  All axes are dimensionless.}
    \label{fig:multiqubitavg}
\end{figure}
\begin{table}[h]
    \centering
    \renewcommand{\arraystretch}{1.5} % Increase row height
    \begin{tabular}{|p{2cm}|p{2cm}|p{2cm}|p{2cm}|p{3cm}|}
\hline
$q$ & $\langle \mathcal{\tilde F}_\theta^q \rangle / \tilde N$  & $\tilde N$ & $\max_N{\langle\mathcal{ F}_\theta^q\rangle}$ & $N$ corresponding  to $\max_N{\langle\mathcal{ F}_\theta^q\rangle}$ \\ \hline
0.95   & 2.17   & 10  & 36.4   & 26   \\ \hline
0.99   & 11.16   &  50  & 956.4   & 133   \\ \hline
\end{tabular}
    \caption{Given two values of $q$, the maximum value of both $\theta$-averaged FI, $\max_N{\langle\mathcal{ F}_\theta^q\rangle}$ and $\theta$-averaged FI per qubit, $\langle \mathcal{\tilde F}_\theta^q \rangle / \tilde N$, and the corresponding number of qubits for each maximum value are listed from Fig.~\ref{fig:multiqubitavg}.}
    \label{tab:multiqubitavg}
\end{table}
From Table~\ref{tab:multiqubitavg}, it is evident that given $N=1000$ and $q=0.99$, it is better to use $20$ number of  $50$-qubit NITDM circuits to get $20\langle\mathcal{\tilde F}_\theta^{0.99}\rangle=11160$ where $Nq^2=980.1$. Similarly, for $q=0.95$ we can use $10$-qubit circuits $100$ times which leads to $100\langle\mathcal{\tilde F}_\theta^{0.95}\rangle=2170$, but $Nq^2=902.5$.
\end{document}